\documentclass[aps,physrev,reprint,superscriptaddress]{revtex4-2}

\usepackage{graphicx}
\usepackage{dcolumn}
\usepackage{bm}
\usepackage{amsmath}
\usepackage{amssymb}
\usepackage{amsthm}
\usepackage{circuitikz}
\usetikzlibrary{decorations.pathreplacing,calligraphy}

\usepackage{thmtools}
\usepackage{thm-restate}
\usepackage{hyperref}
\declaretheorem[name=Proposition,numberwithin=section]{proposition}

\declaretheorem[name=Remark,numberwithin=section]{remark}

\begin{document}

\title{Critical-like growth in non-critical transitions: linearised dynamics and epidemic applications}

\author{Joshua Looker}
\affiliation{EPSRC \& MRC Centre for Doctoral Training in Mathematics for Real-World Systems, University of Warwick, Coventry, United Kingdom}
\author{Kat S. Rock}
\affiliation{The Zeeman Institute for Systems Biology \& Infectious Disease Epidemiology Research, and Mathematics Institute, University of Warwick, Coventry, United Kingdom}
\author{Louise Dyson}
\affiliation{The Zeeman Institute for Systems Biology \& Infectious Disease Epidemiology Research, and Mathematics Institute, University of Warwick, Coventry, United Kingdom}
\affiliation{School of Life Sciences, University of Warwick, Coventry, United Kingdom}

\begin{abstract}
In the wake of the SARS-CoV-2 pandemic, there has been heightened interest from applied mathematicians in infectious disease modelling. Modelling efforts often focus on predicting whether diseases are likely to be eliminated or, instead, (re-)emerge, especially as a result of control measures. This tipping point between elimination and infection waves has been successfully anticipated in the literature through the use of early warning signals and such signals often rely on the theory of critical slowing down. Recent developments have shown that these signals (increases in fluctuation variance and return time) can emerge from the system geometry in the case of non-normal dynamics rather than a change in asymptotic stability. We show how such dynamical behaviour occurs in the fluctuations from the mean-field in general stochastic systems. Using the susceptible-infectious-recovered model as an example application, we analyse how critical-like behaviour can be exploited to anticipate infection waves in the absence of an equilibrium bifurcation.
\end{abstract}

\maketitle

\section{Introduction}
The theory of early warning signals (EWS) has shown increasing promise in anticipating bifurcations in infectious disease models. In the literature, analysis generally considers the transcritical bifurcation between disease-free and endemic states as the basic reproduction number ($R_0$) is varied \cite{brett_anticipating_2017,brett_anticipating_2018,Looker_identifying_2025,southall_early_2021,southall_prospects_2020,gama_dessavre_problem_2019,ghadami_anticipating_2023,Drake_statistics_2019}. Near this bifurcation point, relaxation rates of stochastic perturbations slow down (known as critical slowing down, CSD), leading to identifiable trends in time series statistics, such as an increase in variance \cite{kuehn_mathematical_2011}. However, growing research for other dynamical systems has revealed that similar statistical patterns can occur in systems that are linearly stable and far from any bifurcation \cite{Troude_pseudo_2025, ORegan_transient_2020}. These effects, sometimes termed pseudo-bifurcations or transient critical-like behaviour, arise from the geometry of the underlying dynamics and transient behaviour, rather than from changes in asymptotic stability. For epidemiological models, there are also recent results in the literature highlighting that these transient effects can also be used to identify points of interest in the time series of reported cases (both in terms of the standard bifurcation analysis on approach to (re-)emergence or elimination, and peaks and troughs in the incidence) \cite{proverbio_performance_2022,dablander_overlapping_2022,kaur_anticipating_2020,Looker_identifying_2025}.

A key concept underlying these phenomena is non-normality \cite{Trefethen_spectra_2005}. A linear operator is non-normal when it does not commute with its hermitian transpose, meaning its eigenvectors are not orthogonal (i.e., its commutator, $[A,A^*]:=AA^*-A^*A\neq 0$). In such systems, perturbations can interact through the non-orthogonal eigenbasis, leading to transient amplification even when all eigenvalues indicate asymptotic decay. The short-term growth of perturbations, often referred to as transient amplification or `pseudo-bifurcation' behaviour, is driven by the geometric structure of the operator rather than by its spectrum. This is well known in fluid mechanics, where non-normal operators are associated with transient energy bursts and amplified responses to small perturbations, originally formalised by Farrell as generalised stability theory \cite{Farrell_stability_1996a,Farrell_stability_1996b}. More recently, these ideas have been formalised in the context of stochastic dynamical systems, showing that apparent EWS then occur due to sensitivity to the transient geometry of the system flow. Epidemic models are strongly non-normal due to the asymmetric coupling between susceptible and infected compartments. Further, competing time scales related to the rate of infection, recovery period and demography can lead to shearing of the phase-space. Real-world epidemic systems also exhibit stochasticity. In systems with complex eigenvalues, noise can produce stochastic resonance of the existing eigenmodes in the system. Such stochastic resonance has already been identified in fluctuations around the endemic and disease-free equilibria, leading to a rich variety of dynamical behaviour \cite{Rozhnova_lag_2012,Rozhnova_cities_2011,Rozhnova_seasonal_2010,Rozhnova_oscillations_2009,Drake_statistics_2019,oregan_theory_2013,southall_early_2021} and occurs even without non-normality.

Distinct from this, stochastic amplification can occur due to the spatial geometry and eigenvector non-orthogonality, where the noise feeds energy into transiently growing directional modes regardless of eigenvalue complexity. Non-normality can be examined through the properties of both the Jacobian ($J$) of the system and its symmetric part ($H=(J+J^T)/2$) \cite{Trefethen_spectra_2005}. While the Jacobian determines the asymptotic stability through its eigenvalues, the symmetric part encodes how the norm of perturbations grows or decays instantaneously. Positive eigenvalues of the symmetric part correspond to directions in which perturbation energy increases, even if all eigenvalues of the full Jacobian have negative real parts. This separation between asymptotic stability and instantaneous amplification is fundamental to understanding non-normal transient dynamics. In the context of epidemic models, examining the symmetric part of the Jacobian offers a geometric perspective on when and where the system locally supports transient growth in incidence or prevalence, even without crossing a true bifurcation threshold \cite{ORegan_transient_2020}.

Existing results in the infectious disease dynamics literature investigate both the equilibrium and transient behaviour in EWS as a result of the system geometry. Such analyses generally study the linear probability distribution of the system's fluctuations, often derived via the linear noise approximation (LNA) or by directly linearising the governing stochastic differential equations (SDEs). Although normally expanded around the steady-states of the system, it can also be used to analyse the fluctuations from mean-field solutions. The resulting covariance matrix of these fluctuations evolves according to a time-dependent Lyapunov equation driven by the Jacobian of the deterministic system. In stationary settings, this equation can be solved algebraically to yield steady-state variance spectra, which has been carried out for various disease models \cite{Looker_identifying_2025,Miller_seasonal_2017,ORegan_mosquito_2016,southall_prospects_2020,southall_early_2021,gama_dessavre_problem_2019}. 

However, analysing the linearised fluctuations does have limits when approaching the bifurcation or other exceptional points: the principal eigenvalue of the Jacobian approaches zero, causing the linear restoring forces to vanish. In this immediate critical vicinity, truncated nonlinearities dominate and the linearised covariance evolution ceases to be physically accurate \cite{kuehn_mathematical_2011,Donovan_moment_2024}. Conversely, when the Jacobian varies in time, such as along a system trajectory, the covariance dynamics become explicitly non-autonomous. In such systems, the covariance may increase transiently even if the deterministic trajectory remains within a stable regime, leading to apparent signatures of instability that are purely geometric and transient in origin. Although there are some results in the literature warning that such signatures may reduce EWS applicability, we instead seek to exploit this behaviour to anticipate areas of interest in disease time-series (such as peaks and troughs).

 Rather than proposing a new EWS, this paper aims to provide an analytical foundation for how and when these CSD-like signatures arise from time-varying operator geometry. We build upon the theory of pseudo-bifurcations and generalised stability theory to consider how slowing-down behaviour and related EWS can occur in a general stochastic system under the LNA. We represent the time-varying Lyapunov equation for the covariance matrix in both the Jacobian and symmetric eigenbases to help illuminate how the evolving system geometry impacts the covariance matrix. These reveal contributions arising from the deterministic mode coupling, transient amplification and eigenbasis evolution. Next, we analyse the eigenvalues of the time-dependent Jacobian and its symmetric part for the susceptible-infectious-recovered (SIR) model with demography to identify areas of the SI phase plane where EWS may occur. Finally, we use a set of simulations to test the applicability of this theory to simulated case data and changes in both population and disease demographics.

\section{Time-dependent critical slowing down behaviour}
In this section, we build on the existing literature on pseudo-bifurcations and differential Lyapunov equations by considering a generalised stochastic system and how fluctuations away from a mean-field solution can display critical slowing down-like behaviour.

Let $Y(t)$ be a stochastic variable obeying a stochastic differential equation (such as that given by the chemical Langevin equation) with corresponding drift/deterministic evolution given by the set of differential equations in the form 
\begin{equation}
    \dot{y}=f(y,t),\quad y_0=y(t_0)
\end{equation}
where $y(t)$ is the state of the deterministic system and $t_0$ the initial time. Let $\bar{y}(t)$ be the exact, time-dependent mean-field trajectory solved from these equations. We can then linearise $Y(t)$ around $\bar{y}(t)$, resulting in the stochastic fluctuation process $X(t)=Y(t)-\bar{y}(t)$.

To analyse the stability of these fluctuations, we consider the purely deterministic limit of these stochastic perturbations, denoted $x(t)$, which evolves according to the linearised ordinary differential equation:
\begin{equation}
    \dot{x}=J(t)x,\quad x_0=x(t_0)
\end{equation}
where $J(t)$ is the Jacobian matrix of $f(y,t)$ evaluated at $\bar{y}(t)$ at time $t$. We emphasise that $J(t)$ is time-dependent because it is evaluated along the non-stationary trajectory $\bar{y}(t)$. Therefore, even for autonomous theoretical models, the linearised fluctuations is strictly non-autonomous for transient perturbations/when the system is not at a deterministic equilibrium. We note that these equations only hold whilst the perturbations are sufficiently close to the mean-field behaviour such that a linear approximation to the dynamics is appropriate. 

We now take concepts from the theory of non-autonomous linear equations and linear control theory that are well-known in their respective fields but, to our knowledge, are not well reported in the EWS literature (even in recent papers on pseudo-bifurcations) \cite{abou-kandil_matrix_2003,Behr_solutions_2019}.

Solutions to this set of equations can be written as 
\begin{equation}
    x(t)=\Phi(t,t_0)x(t_0),\quad \dot{\Phi}=J\Phi,\quad\Phi(t,t)=\bf{I}
\end{equation}
where $\Phi$ is the state-transition matrix (also called the fundamental matrix) of the linearised system. Noting that $J$ is a square matrix, and letting $||\cdot||$ be an induced matrix norm, the associated logarithmic norm of $J$ is defined as 
\begin{equation}
    \mu(J)=\lim_{h\rightarrow 0^+}{\dfrac{||{\bf I}+hJ||-1}{h}}.
\end{equation}

We can use this to show how the time-dependent eigenvalues of $J$ lead to CSD in the fluctuations as their real part approaches zero from below. Define also $\mu_*(J)=-\mu(-J)$ and using the upper ($\dfrac{\mathrm{d}}{\mathrm{d}t^+}(x):=\limsup_{h\rightarrow 0^+}{\dfrac{x(t+h)-x(t)}{h}}$) and lower ($\dfrac{\mathrm{d}}{\mathrm{d}t^-}(x):=\liminf_{h\rightarrow 0^+}{\dfrac{x(t)-x(t-h)}{h}}$) right Dini derivatives we derive exponential bounds on $||x||$. Although, a well-known result, we provide the full derivation in Appendix~\ref{app:gini} with the upper bound on the magnitude of the fluctuations given as
\begin{equation}
    \dfrac{\mathrm{d}}{\mathrm{d}t^+}\ln{||x(t)||} \le \mu(J).
\end{equation}

For $||\cdot||=||\cdot||_2$ we have
\begin{align}
    \dfrac{\mathrm{d}}{\mathrm{d}t^+}\ln{||g(t)||} &\le \lambda_{\max}\left(\dfrac{J+J^T}{2} \right)\\
    ||\Phi(t,t_0)||_2 &\le \exp{\left(\int_{t_0}^{t}{\lambda_{\max}\left(\dfrac{J+J^T}{2}\right)ds} \right)}
\end{align}
such that the instantaneous bounds on the growth of fluctuations are given by the maximal eigenvalue of the symmetric part of $J$, $H=(J+J^T)/2$, and the actual magnitude is bounded by Gronwall's inequality. As such, we only expect growth in $x$ if the spectral radius of $J+J^T$ is positive.

Throughout this paper, we use the 2-norm as it is the unique $L^p$ norm defined by an inner product. This allows for the decomposition of the Jacobian into its symmetric and skew-symmetric parts and to utilise an orthogonal spectral geometry (which other norms do not support). As the variance is inherently an $L^2$ concept, using the 2-norm also allows us to relate it to the energy of the system. Finally, we do note that other norms, such as the total variation (TV), could be used. However, this distance is bounded by the Kullback-Leibler divergence, which, under the linear approximations used in the manuscript, is entirely determined by the mean and the covariance matrix. Therefore, using bounds derived via the 2-norm already captures the mechanisms that would drive a divergence in the TV norm.

We note existing work on the definitions of pseudo-bifurcation and non-normal CSD where the time-independent and equilibrium case was considered in \cite{Troude_pseudo_2025}, with a framework for generalised system dynamics also presented in \cite{Troude_unifying_2025}. The remainder of this study aims to extend these studies to show how the linearised dynamics lead to predictable changes in the covariance matrix in the non-autonomous case and how these geometric effects can be used to anticipate areas of interest in infections data.

\subsection{Variance-based indicators in practice}
We continue by considering how the variance and other indicators of transient dynamics might behave in stochastic realisations of epidemic models and dynamical systems in general.

Consider a linearised system corresponding to fluctuations away from a mean-field solution. Let $X(t)\sim N(0,\Sigma_t)$ be the deviations from the mean-field solution $\mu(t)$, as previously defined, then the SDE for the deviations is given by 
\begin{equation}
    dX_t = J(t)X_tdt+G(t)dW_t
\end{equation}
where $B(t)=G(t)G(t)^T$ is the noise-covariance matrix of the transitions in the system (commonly derived from the LNA applied to the master equation of $Y$) and $W_t$ represents the standard Wiener process. We note that this equation can be derived by linearising the stochastic representation of $Y$ such as from the chemical Langevin equation. Then the evolution of the covariance is given by the continuous-time-varying Lyapunov equation below 
\begin{equation}\label{eq:lyapunov}
    \dfrac{\mathrm{d}\Sigma_t}{\mathrm{d}t} = J(t)\Sigma_t+\Sigma_tJ^T(t)+B(t)
\end{equation}
and the fluctuations are a (multi-dimensional) Ornstein-Uhlenbeck process with time-varying coefficients. This equation has an analytic solution (known in control theory and SDE literature \cite{abou-kandil_matrix_2003,Behr_solutions_2019,Sarkka_Applied_2019} but to our knowledge not applied in an EWS context) in terms of the fundamental matrix $\Phi$ (as defined previously) by considering the similarity transform $Z(t)$ given by 
\begin{equation}
    Z(t)=\Phi(t)^{-1}\Sigma_t\Phi(t)^{-T},
\end{equation}
where $(\cdot)^{-T}$ denotes the inverse transpose of a matrix. By considering that $\Phi\Phi^{-1}=\bf{I}$, we also derive identities for $\dfrac{\mathrm{d}}{\mathrm{d}t}({\Phi}^{-1})$ and $(\dfrac{\mathrm{d}}{\mathrm{d}t}(\Phi^{-1}))^T$:
\begin{align}
    \dfrac{\mathrm{d}}{\mathrm{d}t}(\Phi\Phi^{-1}) &= \dfrac{\mathrm{d}}{\mathrm{d}t}(\bf{I}), \nonumber\\
    \dfrac{\mathrm{d}}{\mathrm{d}t}({\Phi})\Phi^{-1}+\Phi\dfrac{\mathrm{d}}{\mathrm{d}t}(\Phi)^{-1} &= 0, \nonumber\\
    \implies \dfrac{\mathrm{d}}{\mathrm{d}t}(\Phi^{-1}) &= -\Phi^{-1}\dot{\Phi}\Phi^{-1} =-\Phi^{-1}J, \nonumber\\
    (\dfrac{\mathrm{d}}{\mathrm{d}t}({\Phi}^{-1}))^T &= (-\Phi^{-1}J)^T=-J^T\Phi^{-T}.
\end{align}

Then differentiating $Y(t)$ gives (dropping the dependence on $t$ for brevity)

\begin{align}\label{eq:integral-cov}
    \dot Z &= -\Phi^{-1}J\Sigma\Phi^{-T} + \Phi^{-1}(J\Sigma+\Sigma J^T+B)\Phi^{-T} \nonumber\\
    &\qquad\qquad\qquad - \Phi^{-1}\Sigma J^T\Phi^{-T} \nonumber\\
    &= \Phi^{-1}B\Phi^{-T} \nonumber\\
&\implies Z(t) =Z(0)+\int_0^t{\Phi(s)^{-1}B(s)\Phi(s)^{-T}ds} \nonumber\\
&\implies \Sigma_t=\Phi(t)\left[\Sigma_0+\int_0^t{\Phi(s)^{-1}B(s)\Phi(s)^{-T}ds}\ \right]\Phi(t)^T\nonumber \\
\end{align}
where the last equality comes from noting that $\Phi(0)=0$ such that $Z(0)=\Sigma_0$. As previously stated, evolution equations in this form are known in other areas but have been reproduced here for the first time in the EWS literature. The remainder of this section represents our novel contributions to fluctuation theory.

We can also decompose our continuous-time Lyapunov equation for the covariance matrix further, and consider explicit bounds on the perturbations from the SDE. Note that for brevity, notation of the time-dependence may be dropped in the remainder of this paper (e.g. $J(t)$ would be written as $J$ and $\Sigma_t$ as $\Sigma$). 

We first consider the effects of $H$ on the total variance by taking the trace of Eq.~\ref{eq:lyapunov} yielding
\begin{align}\label{eq:lyapunov-ineq}
    \dfrac{\mathrm{d}\Sigma}{\mathrm{d}t} &= J\Sigma+\Sigma J^T+B \nonumber \\
    \dfrac{\mathrm{d}}{\mathrm{d}t}\operatorname{tr}\Sigma &= 2\operatorname{tr}(H\Sigma)+\operatorname{tr}B \nonumber\\
    \implies &\dfrac{\mathrm{d}}{\mathrm{d}t}\operatorname{tr}\Sigma\le 2\lambda_{\max}(H)\operatorname{tr}\Sigma + \operatorname{tr}B
\end{align}
where the first line comes from the cyclicity of the trace and the inequalities come from the fact that $\Sigma \succeq 0$. As $B$ is a noise-covariance matrix, we also know that $B\succeq 0$, and so this noise term continuously injects variance into the system with the deterministic contribution of $H$ dampening or amplifying these fluctuations depending on its spectral properties (as expected from CSD theory \cite{ORegan_transient_2020}). If $\lambda_{\max}(H)<0\implies\operatorname{tr}(H\Sigma)<0$, then there will be deterministic linear damping of the variance (and vice versa for $\lambda_{\max}(H)>0$). Therefore, a necessary condition for instantaneous linear growth of the variance is if $\lambda_{\max}(H)>0$ (with the instantaneous sufficient condition being $2\lambda_{\min}(H)+\operatorname{tr}B>0$ which, since $B\succ 0$, reduces to $H\succeq 0$).

Further, we can also separate the effects of `time-dependent CSD' defined in the normal way as a change in the real part of the dominant eigenvalue(s) of the Jacobian matrix from the transient effects caused by the eigenmodes of $H$.

We assume here that $J(t)$ is diagonisable, with time-dependent right and left eigenvectors forming a biorthonormal eigenbasis and corresponding eigenvalues $\lambda_i$, such that $J(t)=V(t)\Lambda(t)V(t)^{-1},\Lambda(t)=\operatorname{diag}(\lambda_1,\dots,\lambda_n)$, where $V(t) = [v_1,\dots,v_n]$ is the matrix of right eigenvectors and $W(t)=[w_1,\dots,w_n], W^*V={\bf I}$ is the left eigenvectors, normalised to be biorthogonal. Note that this assumption holds for many systems but fails when the dynamics have unidirectional coupling or at exceptional points (such as at the bifurcation). We define the modal covariance and noise to be $\Sigma_y := V^{-1}\Sigma\,V^{-*},B_y := V^{-1}B\,V^{-*}$ and set $M(t):=V(t)^{-1}\dot V(t)=$. We note that $(\cdot)^*$ denotes the Hermitian transpose and arrive at the following proposition:

\begin{proposition}[Modal ($J$-eigenbasis) decomposition]\label{thm:modal}
Under the above definitions and assumptions, the following matrix equality holds
\begin{equation}\label{eq:modal_matrix}
\dot\Sigma_y \;=\; \Lambda\Sigma_y + \Sigma_y\Lambda^* + B_y \;-\; M\Sigma_y \;-\; \Sigma_y M^*,
\end{equation}
where \((\cdot)^*\) denotes the Hermitian transpose, and, for the diagonal modal variances $p_i(t):=(\Sigma_y)_{ii}$, the exact scalar ordinary differential equation (ODE) is
\begin{equation}\label{eq:modal_scalar}
\dot p_i \;=\; 2\Re(\lambda_i)\,p_i \;+\; (B_y)_{ii}
\;-\; 2\Re\left(\sum_{b} M_{ib}{\Sigma_y}_{bi}\right)
\end{equation}
for each \(i=1,\dots,n\).

\smallskip\noindent\textbf{Remark.} In \eqref{eq:modal_scalar}, the first term $2\Re(\lambda_i)p_i$ is the pure modal growth/decay; $(B_y)_{ii}$ is the modal noise injection; the remaining sum encodes the time-evolution of the modal basis.
\end{proposition}

\begin{proof}
Differentiate \(\Sigma_y = V^{-1}\Sigma V^{-*}\):
\begin{equation*}
    \dot\Sigma_y = \dot V^{-1}\Sigma V^{-*} + V^{-1}\dot\Sigma V^{-*} + V^{-1}\Sigma \dot V^{-*}.
\end{equation*}
From the definition $V^{-1}V={\bf I}=V^{-*}V^*$ we obtain the identities $\dot V^{-1} = -V^{-1}\dot V V^{-1} = - M V^{-1}$ and $\dot V^{-*} = -(V^{-*}\dot V^* V^{-*})=-(MV^{-1})^*$ to then obtain
\begin{equation*}
    \dot\Sigma_y = -M\Sigma_y - \Sigma_y M^* + V^{-1}\dot\Sigma V^{-*}.
\end{equation*}
Substitute $\dot\Sigma = J\Sigma + \Sigma J^T + B$, and use $V^{-1}J = \Lambda V^{-1}$ and its hermitian transpose to get $V^{-1}\dot\Sigma V^{-T} = \Lambda\Sigma_y + \Sigma_y\Lambda^T + B_y$. Rearranging and simplification of the expanded expression yields \eqref{eq:modal_matrix} and, noting that $\Sigma_y$ is hermitian, extracting diagonal entries of $\dot\Sigma_y$ gives \eqref{eq:modal_scalar}.
\end{proof}

In the eigenmode basis, it is clear then that the degree of modal growth depends on the real parts of the eigenvalues of $J$ (as expected) and the projection of $B$ in the eigenspace of $J$. Intuitively, these equations decompose the covariance evolution (in the eigenmode basis), into contributions from the Jacobian matrix, the injected noise and the rate-of-change of the eigenbasis itself. 

\smallskip\noindent\textbf{Remark.} We also note that at stationarity, such that $\dot{\Sigma}_y=M=0$ we regain results previously derived in \cite{chen2019} as $(\Sigma_y)_{ij}=-(B_y)_{ij}/(\lambda_i+\lambda_j)$.

Finally, we must emphasise that these equations only hold under the assumption of a diagonalisable Jacobian. While this may be a standard assumption for many dynamical disease systems, it fails in two key scenarios. It fails persistently, in systems with specific topological constraints (such as unidirectional compartmental models), and locally, at exceptional points where eigenvalues collide.

In these regimes, the system exhibits an extreme increase in non-normality, leading to a highly ill-posed response to perturbations \cite{Santos_perturbations_2022}. As the eigenvectors align, $V$ becomes nearly singular and $B_y,M$ become extremely sensitive to small fluctuations. Rather than standard exponential behaviour, defective dynamics force perturbations to undergo polynomial-exponential transient amplification (scaling as $t^k e^{\Re(\lambda)t}$, where $k$ is the order of the degeneracy and the form is due to the generalised eigenvalues/Jordan normal structure), representing the theoretical maximum of non-normal sensitivity.

Nonetheless, it is possible to extend the above derivation to a defective matrix through generalised eigenvectors and the Jordan normal form. Omitting the algebra, the final form is
\begin{equation}
    \dot{\Sigma}_y=(\Lambda+N)\Sigma_y+\Sigma_y(\Lambda+N)^*+B_y-M\Sigma_y-\Sigma_yM^*,
\end{equation}
where $J = V(\Lambda + N)V^{-1}$ represents the Jordan normal form of the defective matrix. Here, $V$ is the (generally non-unitary) matrix of generalised eigenvectors, $\Lambda$ is the diagonal matrix of eigenvalues, and $N$ is a strictly upper-triangular nilpotent matrix (satisfying $N^k = \mathbf{0},k\le n$). Specifically, $N$ contains 1s on the super-diagonal corresponding to the coupled generalised eigenvectors of each Jordan block, and zeros elsewhere. Then, the derivation continues in the exact same way as in deriving \eqref{eq:modal_matrix}, replacing $\Lambda$ with $\Lambda+N$. The scalar equation for the diagonal modal variances behaves identically to \eqref{eq:modal_scalar} with the real part of the eigenvalues replaced by the real part of the generalised eigenvalues. Note that the form of the generalised eigenvectors identifies the polynomial-exponential transient amplification mentioned previously.

This formulation allows us to explicitly predict the transient variance growth in defective systems. The presence of the nilpotent matrix $N$ encodes the non-normal amplification highlighted above. Because $N$ couples adjacent generalised modes, it acts as a deterministic feed-forward mechanism. In the scalar equations for the modal variances, $N$ increases variance from one generalised mode into the next. This structural cascade alters the fluctuation dynamics, leading to polynomial bursts in variance that are completely independent of the baseline noise injection $B$, before exponential decay occurs.

\smallskip\noindent\textbf{Remark.} We note that in the derivation of \eqref{eq:modal_matrix}, we have arrived at an intermediary formula which holds for a general time-varying basis that we will make use of in the next theorem. Explicitly, this evolution is given as
\begin{equation}\label{eq:ev_basis}
\dot{\Sigma}_T=J_T\Sigma_T+\Sigma_TJ_T^*+B_T+\left(\dot{T}T^{-1}\right)\Sigma_T+\Sigma_T\left(\dot{T}T^{-1}\right)^*
\end{equation}
where $T$ is the change of basis matrix such that $Z_T=TZT^*$ for all matrices $Z$ (where we note that setting $T=V^{-1}$ leads to the previously derived equations).

We next consider a decomposition in terms of the eigenbasis of $H$ to identify the conditions for non-normal growth in perturbations and variance. Let $H=\dfrac{1}{2}(J+J^T), K=\dfrac{1}{2}(J-J^T)$ such that $J=H+K$. Further, let $S$ be an orthonormal matrix that diagonalises $H$ with $H=S\Lambda_HS^T$ such that the eigenvalues $\mu$ are all real (as $H$ is a real-symmetric matrix). Define $X_H=S^TXS$ for any matrix $X$ and $M_S=\dot{S}^TS$. We then arrive at the following proposition:

\begin{proposition}[Symmetric ($H$) eigenbasis decomposition]\label{thm:symmetric}
Under the above assumptions and definitions the covariance ODE transforms to
\begin{align}\label{eq:H_basis_matrix}
    \dot\Sigma_H \;&=\; (\Lambda_H\Sigma_H+\Sigma_H\Lambda_H)+(K_H\Sigma_H-\Sigma_HK_H)+B_H \nonumber \\
    &\quad\quad\quad+(M_S\Sigma_H-\Sigma_HM_S)
\end{align}
and the diagonal elements ($\sigma_i=(\Sigma_H)_{ii}$) satisfy the scalar ODE
\begin{equation}\label{eq:H_basis_scalar}
    \dot\sigma_i \;=\; 2\mu_i\,\sigma_i
    \;+\; (B_H)_{ii} \;+\; 2\sum_{k} \big( {K_H}_{ik}+{M_S}_{ik}){\Sigma_H}_{ki}.
\end{equation}

\smallskip\noindent\textbf{Remark.} We note that summing \eqref{eq:H_basis_scalar} and \eqref{eq:modal_scalar} over $i$ leads to the previously derived evolution equation for the trace of $\Sigma$ as the trace is invariant under a change of basis. Thus, these equations intuitively give the contributions of the modal and non-normal dynamics to the total variance in the system.
\end{proposition}

\begin{proof}
Setting $T=S^T$ in \eqref{eq:ev_basis} and noting that $S$ is real such that $(S^T)^*=S$, then taking the expansion of $J_H$ and using the fact that $K_H,M_S$ are skew-symmetric leads directly to \eqref{eq:H_basis_matrix}. Finally, taking the diagonal elements of \eqref{eq:H_basis_matrix} and again noting that $K_H,M_S$ are skew-symmetric to simplify the commutator terms leads to \eqref{eq:H_basis_scalar}.
\end{proof}

To summarise this section; along a trajectory, it is possible to have growth in the variance from the real part of the eigenvalues of the Jacobian, akin to `time-dependent' critical slowing down. This balances the noise injection from $B$ and the transient growth from $H$ during the periods when its eigenvalues are positive.

We now consider what perturbations maximise the modal and non-normal growth in variance. To do this, we will use methods from optimal perturbation theory \cite{Farrell_stability_1996b, neubert_reactivity_2004, Trefethen_spectra_2005,Santitissadeekorn_Optimally_2010}. Defining the system energy via the standard inner product $E(t) = \Vert x(t) \Vert_2^2$, we then make the following proposition (note that similar propositions can be found in \cite{Farrell_stability_1996b,Farrell_Perturbation_II_2002}):

\begin{proposition}\label{prop: optimal_H}
    The principal eigenvector $v_1(t_0)$ of $H(t_0)$ defines both the state perturbation $x_0\propto v_1(t_0)$ and initial covariance matrix $\Sigma_0 = c_0 v_1(t_0)v_1(t_0)^T$, that maximises the instantaneous variance growth rate at time $t_0$ (and the instantaneous energy growth rate $\frac{1}{E(t_0)} \left. \dfrac{\mathrm{d}E}{\mathrm{d}t}\right \vert_{t=t_0}$). Here, for $\Sigma_0\neq 0$ we assume an initial energy budget $E_0=\Vert x_0 \Vert_2^2$ and initial uncertainty $\operatorname{tr}(\Sigma_0)=c_0$).
\end{proposition}

\begin{proof}
    Differentiating the energy at $t=t_0$ and substituting the definition of $H$ we recover the identity
    \begin{equation*}
        \left. \dfrac{\mathrm{d}E}{\mathrm{d}t}\right\vert_{t=t_0} = 2 x^TH(t_0)x.
    \end{equation*}
    The instantaneous fractional energy growth rate is given by the Rayleigh quotient of $H(t_0)$:
    \begin{equation*}
        \sigma_E(t_0) = \frac{1}{E(t_0)}\left.\dfrac{\mathrm{d}E}{\mathrm{d}t}\right\vert_{t=t_0} = \frac{2x^TH(t)x}{x^Tx}.
    \end{equation*}
    As $H(t)$ is Hermitian, by the Courant-Fischer-Weyl min-max principle this is strictly bounded by the maximum eigenvalue:
    \begin{equation*}
        \sigma_E(t)=\frac{2x^TH(t)x}{x^Tx} \leq 2\lambda_{\text{max}}(H(t)).
    \end{equation*}
    Equality holds iff $x(t)$ is collinear with the principal eigenvector $v_1(t)$ corresponding to $\lambda_{\text{max}}H(t))$, such that $x_0\propto v_1(t_0)$ uniquely maximises the instantaneous deterministic growth.

    Returning to Eq.~\ref{eq:lyapunov}, for the initial total uncertainty $c_0$ we maximise the bilinear term $\operatorname{tr}(H(t_0)\Sigma_0)$. Using the previous decomposition, $H=S\Lambda_H S$, and using the cyclic property of the trace:
    \begin{equation*}
        \operatorname{tr}(H\Sigma) = \operatorname{tr}(S\Lambda_H S^T\Sigma)=\operatorname{tr}(\Lambda_H S^T\Sigma S)=\operatorname{tr}(\Lambda_H \Sigma_H).
    \end{equation*}
    We can easily bound this by replacing every $\mu_i$ with the maximum eigenvalue $\mu_1$:
    \begin{equation*}
        \sum_{i=1}^n \mu_i \Sigma_{H_{ii}} \le \mu_1 \sum_{i=1}^n \Sigma_{H_{ii}} = \lambda_{\max}(H(t)) \operatorname{tr}(\Sigma_0).
    \end{equation*}
    As $\Sigma_H\succeq 0$ we know that $\Sigma_{H_{ii}}>0$ such that for equality we require $\Sigma_{H_{11}}=c_0$ and all other elements to be zero. Transforming back to the original basis yields $\Sigma_0 = c_0v_1(t_0)v_1(t_0)^T$.
\end{proof}

Further, this can be extended using $\Phi$ to give the optimal perturbation for a specific time $t$:
\begin{proposition}\label{prop: optimal_Phi}
   For a specific target time $t$, the principal right singular vector $u_1$ of the state-transition matrix $\Phi(t, t_0)$ defines the initial state perturbation $x_0 \propto u_1$ that maximises finite-time energy gain. Similarly, orienting the initial uncertainty entirely along this singular vector ($\Sigma_0 = c_0 u_1 u_1^T$) maximises the deterministic propagation of variance to time $t$.
\end{proposition}
\begin{proof}
    The deterministic state evolves as $x(t) = \Phi(t, t_0)x_0$. The maximum finite-time energy gain is:
    \begin{equation*}
        \max_{x_0} \frac{\Vert{}x(t)\Vert{}_2^2}{\Vert{}x_0\Vert{}_2^2} = \max_{x_0} \frac{x_0^T \Phi(t, t_0)^T \Phi(t, t_0) x_0}{x_0^T x_0}.
    \end{equation*}
    Noting that $\Phi^T\Phi$ is Hermitian, by similar arguments as in Prop~\ref{prop: optimal_H}, the maximum is $\lambda_{\text{max}}(\Phi \Phi^T)$. Taking the singular value decomposition of $\Phi=VDU^T$, where $V$ gives the left-singular vectors, $D$ contains the singular values and $U^T$ are the right-singular vectors, we find that $\Phi^T\Phi =VD^2V^T$. Again, identical logic requires that $x_0\propto u_1$ to maximise the finite-time energy gain.
    
    We next analyse the exact solution for the total system variance, given in Eq.\ref{eq:integral-cov} and expanding the brackets:
    \begin{align*}
        \operatorname{tr}(\Sigma_t) &= \operatorname{tr}\left(\Phi(t, t_0)\Sigma_0\Phi(t, t_0)^T\right)\\
        &\quad+ \int_{t_0}^t \operatorname{tr}\left(\Phi(t, s)B(s)B(s)^T\Phi(t, s)^T\right) \mathrm{d}s.
    \end{align*}
    
    The initial covariance $\Sigma_0$ influences the variance at time $t$ solely through the first term. We can again use the cyclic property of the trace to yield:
    \begin{equation*}
        \operatorname{tr}\left(\Phi(t, 0)\Sigma_0\Phi(t, 0)^T\right) = \operatorname{tr}(D^2V^T\Sigma_0V).
    \end{equation*}
    We can then use identical arguments to Prop.~\ref{prop: optimal_H} to conclude that $\Sigma_0 = c_0u_1u_1^T$. We note however, that these effects may be attenuated by the time-integral with $B(s)$.
\end{proof}

\begin{remark}
    For most real-world systems, $\Phi$ will not be analytically solvable and so would have to be estimated numerically. Thus, Prop.~\ref{prop: optimal_H} is likely more useful when considering specific models rather than data-driven applications.
\end{remark}

\section{SIR model: theoretical results} \label{Section: SIR Jacobian}
Having derived conditions for modal and transient growth in the variance of fluctuations and the magnitude of the fluctuations, we now test the applicability of the theory to a simple infection model. We investigate the conditions in the $SI$ phase plane for the Jacobian eigenvalues to be complex and for the $H$ eigenvalues to be positive. 

\subsection{Jacobian and its symmetric part}
We begin by considering the basic SIR model with demography, representing transmission of an infectious pathogen in a population:
\begin{align}
    \dfrac{\mathrm{d}S}{\mathrm{d}t} &= -\beta SI +\mu(1-S) \\
    \dfrac{\mathrm{d}I}{\mathrm{d}t} &= \beta SI -(\gamma+\mu)I \\
    \dfrac{\mathrm{d}R}{\mathrm{d}t} &= \gamma I - \mu R.
\end{align}
In this case, we take $S,I,R$ to be the proportion in each state and so we note that the system is 2D as $S+I+R=1$ and will instead take these to be integers when simulating stochastic realisations of the model later. The time-dependent Jacobian (for $S,I$) is then
\begin{equation}
    \begin{bmatrix}
        -\beta I-\mu & -\beta S \\
        \beta I & \beta S - \gamma - \mu
    \end{bmatrix}
\end{equation}
with corresponding characteristic polynomial
\begin{equation}
    \lambda^2+\lambda(\beta I-\beta S + \gamma + 2\mu)+(\beta\gamma I+\beta \mu I-\beta\mu S+\gamma\mu+\mu^2)
\end{equation}
where $\lambda$ are the eigenvalues. This polynomial has discriminant
\begin{equation}
    D=\beta^2(I-S)^2-2\beta\gamma(I+S)+\gamma^2
\end{equation}
such that the eigenvalues of the time-dependent Jacobian are complex conjugate when this is less than zero and real otherwise. This is symmetric about the $S=I$ line. We note that in the full system representation (with $R$), the remaining eigenvalue is $\lambda_3=-\mu <0$ such that the asymptotic stability of perturbations is solely dependent on the eigenvalues corresponding to the $SI$ plane.

We note that in this system $R_0=\frac{\beta}{\gamma+\mu}$ is known as the basic reproduction number, the expected number of secondary infections from one infectious individual in an otherwise susceptible population. The system has two fixed points, the disease-free equilibrium (DFE=$(S=1,I=0)$) and the endemic equilibrium (EE=$(S=\frac{1}{R_0},I=\frac{\mu}{\beta}(R_0-1))$) which collide and exchange stability for $R_0=1$ in a transcritical bifurcation. For $R_0\ge1$, the EE exists and is asymptotically stable, where for $R_0<1$ the DFE is stable.

Another quantity of interest to epidemic modellers is the time-varying, effective reproduction number ($R_t$). For simple infection models, this can be defined as $R_t=R_0 S\le R_0$ and we observe that $R_t>1$ if, and only if, $\frac{\mathrm{d}I}{\mathrm{d}t}>0$. Although not an equilibrium bifurcation, this system-state-dependent quantity dictates whether the epidemic instantaneously grows or decays, and transitions through $R_t=1$ have previously been anticipated in the EWS literature \cite{Looker_identifying_2025}. 

We now seek to investigate the behaviour of the manifold $D=0$ as $R_0$ is varied, since the number of times trajectories cross this manifold is the number of times the time-dependent eigenvalues switch complexity. WLOG we set $I=0$ to consider the number of intersections of the manifold with the $S$ axis (in the $SI$ plane). Setting $I=0$ and solving D=0, while noting that $S,\gamma,\beta,\mu,R_0\ge 0$ yields one intersection with the $S$ axis at 
\begin{align}
     S^D=\dfrac{\gamma}{\beta}=\dfrac{\gamma}{R_0(\gamma+\mu)}.
\end{align}
For $R_0 \ge 1$, we have
\begin{equation}
    S^D=\dfrac{\gamma}{R_0(\gamma+\mu)} \le \dfrac{\gamma}{\gamma+\mu}<1
\end{equation}
and so $S^D$ is within the biologically feasible domain of $S$. As $D$ is symmetric in $S,I$, we also find $I^D=S^D$, where $I^D$ is the intersection of $D$ with the $S=0$ axis. Turning to the $R_0<1$ case we still have $S^D=I^D=\frac{\gamma}{\beta}$,
but we note that this is no longer restricted to within the domain of $S,I\in[0,1]$ and is independent of $\mu$ (depending only on the characteristics of the disease). For the intersection point to be within the feasible region we require $\gamma<\beta$, while also require $\mu$ to be sufficiently large such that $\beta<\gamma+\mu$ still holds (so that $R_0<1$). Intuitively, this only applies for diseases in populations where the infection process are of similar time scales as the natural demography.

In either $R_0$ case, as $D$ is quadratic in $S,I$, and $S, I$ have a monotonically decreasing relationship, it cannot change convexity in the region. So, the number of times a trajectory can cross the manifold depends on $\frac{\gamma}{\beta},R_0$ and the gradient of $\frac{\mathrm{d}I}{\mathrm{d}t},\frac{\mathrm{d}S}{\mathrm{d}t}$ along the manifold.

We wish to show that, for biologically reasonable initial conditions, the system will cross this manifold at least twice for $R_0 \ge 1$ and a maximum of once for $R_0<1$.

To simplify the analysis, we will consider each case of $R_0$ individually. We begin with $R_0\ge1$ such that the EE is globally asymptotically stable, and the DFE is unstable. First, we look at just the dependence of $D$ on $S$, there are two parts of the $D=0$ manifold (on either side of $S^D=\frac{\gamma}{\beta}$ we will label $D^+,D^-$). We need to show where $\frac{\mathrm{d}I}{\mathrm{d}t}>0$ along $D^+$ as, with the EE being globally asymptotically stable, the trajectory will then eventually cross over $D^-$ (note this is due to the EE being to the right of the intersection point). We note that as the EE is a stable spiral (in this $R_0$ regime), the trajectories that do cross over $D^+$ will continuously cross over $D^+$ and then $D^-$ every time they oscillate inwards to the EE. However, since $\frac{\mathrm{d}I}{\mathrm{d}t}=0\iff S=\frac{1}{R_0} \geq S^D$, there is a wedge where the trajectory cannot cross upwards across $D^+$ (see Fig.~\ref{fig:placeholder}).

\begin{figure}[!h]
    \centering
    \includegraphics[width=0.65\linewidth]{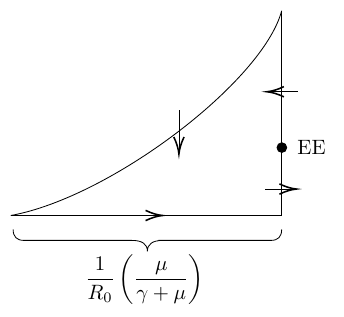}
    \caption{Phase diagram for the `wedge' where the black dot indicates the EE and arrows indicate the direction of the flow.}
    \label{fig:placeholder}
\end{figure}

We can see that this is infinitesimal in size for $\gamma \gg \mu$ such that for diseases where the natural demography is on a much larger time scale than outbreaks, system trajectories are still likely to cross over $D^+$. For diseases that do start close to this trapping region, or in populations with similar $\mu,\lambda$, the EE behaves more like a stable node and trajectories starting close to this wedge will not cross $D$ (from below).

Turning to the $R_0<1$ case, a similar derivative analysis shows that
\begin{equation}
    \dfrac{\mathrm{d}I}{\mathrm{d}t}=I(\gamma+\mu)(R_0 S-1)\le(\gamma+\mu)(R_0-1)<0,
\end{equation}
where the last inequality comes from $R_0<1$. So, regardless of the value of $\frac{\gamma}{\beta}$, we expect no crossings unless the system starts above the $D=0$ manifold (which is unlikely for most real-world systems).

We note that the relationship between $I^D$ and $I^*=\frac{\mu}{\beta}(R_0-1)$ was not considered as, apart from the behaviour near the wedge, whether a system trajectory crosses the manifold can be almost purely determined by the above analysis unless the system starts in a biologically unreasonable area with $I\approx 1$. Even in this case, the symmetry of the manifold with respect to S and I leads to identical conclusions on the relationship between the number of crossings and $R_0$.

\begin{table}
\caption{\label{tab:params}}
\begin{ruledtabular}
\begin{tabular}{lcccc}
& $R_0$ & $\beta$ (days$^{-1}$) & $1/\gamma$ (days) & $1/\mu$ (years)\\
High $R_0$, low $\mu$ & 4 & 0.2859  & 14  & 80  \\ 
High $R_0$, high $\mu$ & 4 & 0.2859 & 14 & 1/73  \\
Low $R_0$, low $\mu$ & 0.840 & 0.06  & 14  & 80 
\end{tabular}
\end{ruledtabular}
\end{table}

To illustrate this analysis numerically, we plot the SI phase plane in Fig~.\ref{fig:SIR} for the parameter values given in the top row of Table~\ref{tab:params}. The colouring of the heat-map corresponds to the discriminant of $J(t)$, with the curved and dotted line the $D=0$ manifold such that we expect normal mode growth in fluctuations as it is approached from the red region.

\begin{figure}[h!]
    \centering
    \includegraphics[width=\linewidth]{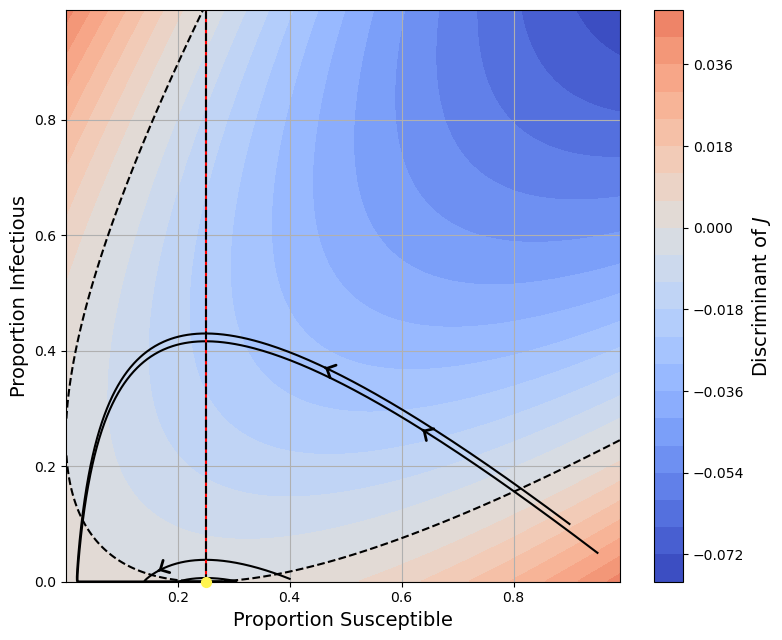}
    \caption{Phase plane of the SIR model with $R_0=4,1/\mu=80\text{ years}$ and heat map denoting the discriminant of the Jacobian. The curved dashed black line shows the $D=0$ manifold, the vertical dashed black line corresponds to $R_t=1$ and the near collinear vertical solid red line indicates the $S$ coordinate of the endemic equilibrium. The solid yellow dot indicates the endemic equilibrium.}
    \label{fig:SIR}
\end{figure}

\newpage
We can see that the trajectories (shown by the black lines) indeed cross the manifold at least twice for this set of parameter values, as expected from the theory. Although the wedge width is negligible for these specific contagion parameters, it is still noticeable between $S^D$ (the vertical black line) and $S^*=\frac{1}{R_0}$ (the vertical red line). In Fig.~\ref{fig:SIRhighmu}, we can see the effects of having a biologically-unreasonably high $\mu$ value; the wedge is much larger and for many reference trajectories we would not expect asymptotic growth in perturbations.

\begin{figure}[h]
    \centering
    \includegraphics[width=\linewidth]{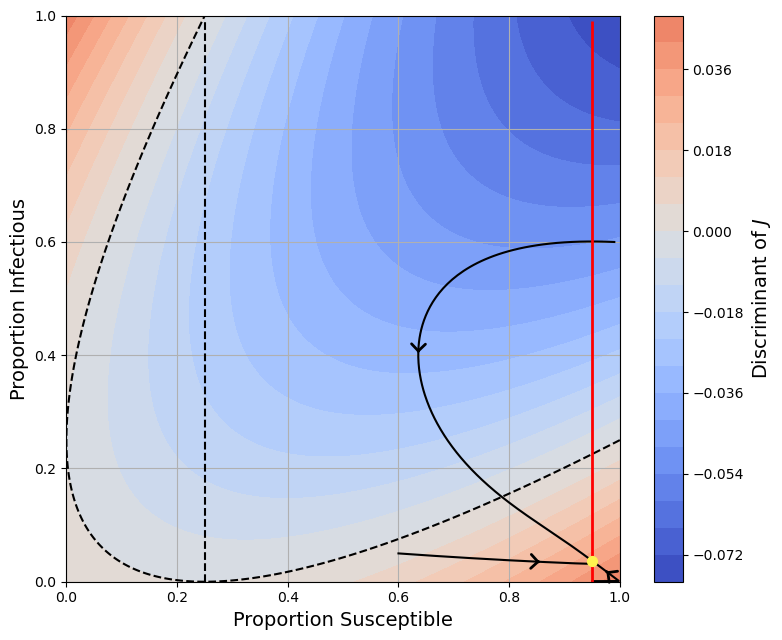}
    \caption{\label{fig:SIRhighmu} Phase plane of the SIR model with $R_0=4,1/\mu=1/73\text{ years}$ and heat map denoting the discriminant of the Jacobian. The curved dotted black line shows the $D=0$ manifold, the vertical dotted black line corresponds to $R_t=1$ and the vertical solid red line indicates the $S$ coordinate of the endemic equilibrium. The solid yellow dot indicates the endemic equilibrium.}
\end{figure}

Finally, we plot the phase plane for the case where $R_0<1$ and $\mu$ is biologically reasonable in Fig.~\ref{fig:SIRlowr0}. The trajectories also clearly do not cross the manifold unless they start above it, so perturbations are expected to eventually decay unless there is a high disease prevalence.

\begin{figure}[h]
    \centering
    \includegraphics[width=\linewidth]{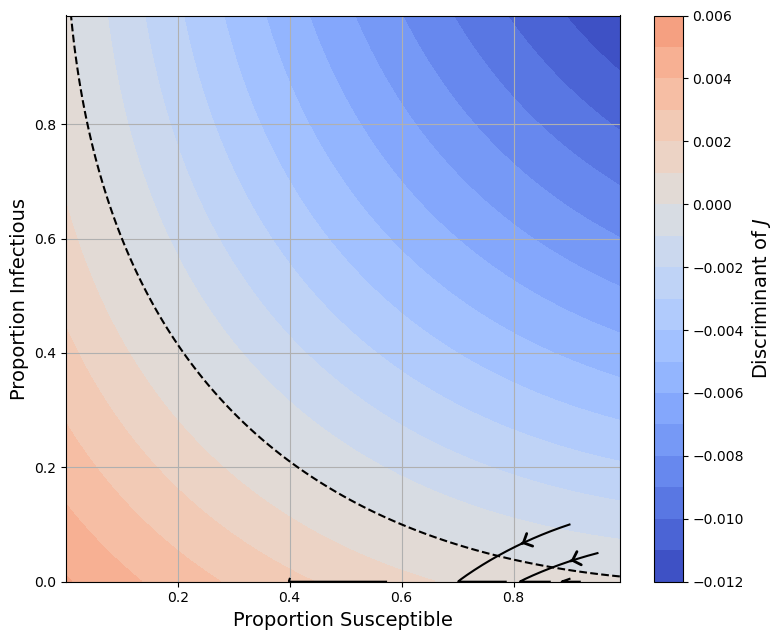}
    \caption{Phase plane of the SIR model with $R_0=0.840,1/\mu=80\text{ years}$ and heat map denoting the discriminant of the Jacobian. The curved dotted black line shows the $D=0$ manifold and there is no endemic equilibrium.}
    \label{fig:SIRlowr0}
\end{figure}

We now turn to the conditions for non-normal/transient growth in the perturbations and associated variance. As stated, non-normal operators can lead to transient growth in transient fluctuations of dynamical systems. As previously stated, non-normal matrices are matrices that do not commute with their Hermitian transpose. For infectious disease compartmental models, this corresponds to a non-normal Jacobian, where the transient dynamics relate to the symmetric part of the Jacobian ($H$) as analysed in the previous section.

Indeed, if we look at this simple model, we find that the commutator, $[J,J^T]$, is 
\begin{equation}
\begin{bmatrix}
        \beta^2(S^2-I^2) & -(\beta S+\beta I)(\beta I+\beta S - \gamma)\\
        -(\beta S + \beta I)(\beta S+\beta I-\gamma) & \beta^2(I^2-S^2)
    \end{bmatrix},
\end{equation}
and is non-zero unless $S=I=0$. 

We now consider the eigenvalues of $H$ to assess the possibility of non-modal growth in transient dynamics too. We first have:
\begin{equation}
   H = \begin{bmatrix}
        -(\beta I+\mu) & \beta (I -S)/2 \\
        \beta(I-S)/2 & (\beta S-\gamma-\mu)
    \end{bmatrix}
\end{equation}
which, from the Sylvester criterion, cannot be positive definite for positive parameters. Thus, the eigenvalues have either opposite signs or are both negative. H has eigenvalues
\begin{align}
 \lambda^{\pm} &= \dfrac{\beta (S-I)-(2\mu+\gamma)\pm\sqrt{\alpha}}{2}
\end{align}
where \begin{equation}
    \begin{array}{lcl}
\alpha &=& (\beta(I-S)+2\mu+\gamma)^2+\beta^2(I+S)^2\\
&&\qquad -4(\mu(\mu+\gamma)+\beta\mu(I-S)+\beta\gamma I).
\end{array}
\end{equation}
We note that, as the first leading minor is negative, $\lambda^-$ is strictly negative. The sign of $\lambda^+$ is then given by the relationship $\lambda^+\lambda^-=\det H$ so that $\text{sgn}(\lambda^+) = -\text{sgn}(\det H)$. Transient growth in fluctuations is possible for $\lambda^+>0$. We can conduct a similar analysis as for the $D=0$ manifold by considering where $H_D=\det H=0$. Again, restricting to the $I=0$ case we have the intersection given by 
\begin{align}
 S_{DH} := \dfrac{-\mu/\beta\pm{\sqrt{(\mu/\beta)^2+4\mu /R_0}}}{2},
\end{align}
where one root is always negative, and we can show that the other root is always less than or equal to $S^*=\frac{1}{R_0}$. Indeed, we can show that if $S\ge \frac{1}{R_0}$ then $\det H <0$ and so $S<\frac{1}{R_0}$ on $H_D$. Considering $\det H$ again, we have 
\begin{align}
    \dfrac{\partial }{\partial S}(\det H) &= -\beta \mu-\beta^2(S+I)/2 <0 \nonumber\\
    \det H|_{S=1/R_0} &=-\dfrac{(\gamma+\mu)^2}{4}(1-IR_0)^2\le0,
\end{align}
so $\det H$ is strictly decreasing for $S\in[0,1]$ and $\det H|_{S=\frac{1}{R_0}}\le 0$ such that the $H_D$ manifold does not exist past $S=\frac{1}{R_0}$. Although not necessarily illuminating, we can also find that the optimal perturbation (the normalised eigenvector associated with this eigenvalue) is given as 
\begin{equation}
    {v}_{1}=\frac{1}{\sqrt{\chi}}\left(\begin{matrix}-\beta (S+I)-\gamma -\sqrt{\alpha }\\ \beta (I-S)\end{matrix}\right)
\end{equation} 
where $\chi = (\beta (S+I)+\gamma +\sqrt{\alpha })^{2}+\beta ^{2}(I-S)^{2}$. Returning to the manifold analysis, we have that $S_{DH}\le S^*$ and, similar to the analysis of the $D=0$ manifold, we would expect trajectories to cross $H_D$ at least twice if the system starts sufficiently away from the EE for $R_0\ge 1$. Notably, if the system crosses $H_D$ at all, it must cross again via the above bound on its S domain.

This is illustrated in Fig.~\ref{fig:SIRnn}, where trajectories starting with a sufficient proportion susceptible or infectious cross the manifold twice. We note that in the following phase planes, a red colouring indicates that $\det H <0$ such that $\lambda^+>0$ and transient growth in fluctuations is possible.

\begin{figure}[h!]
    \centering
    \includegraphics[width=\linewidth]{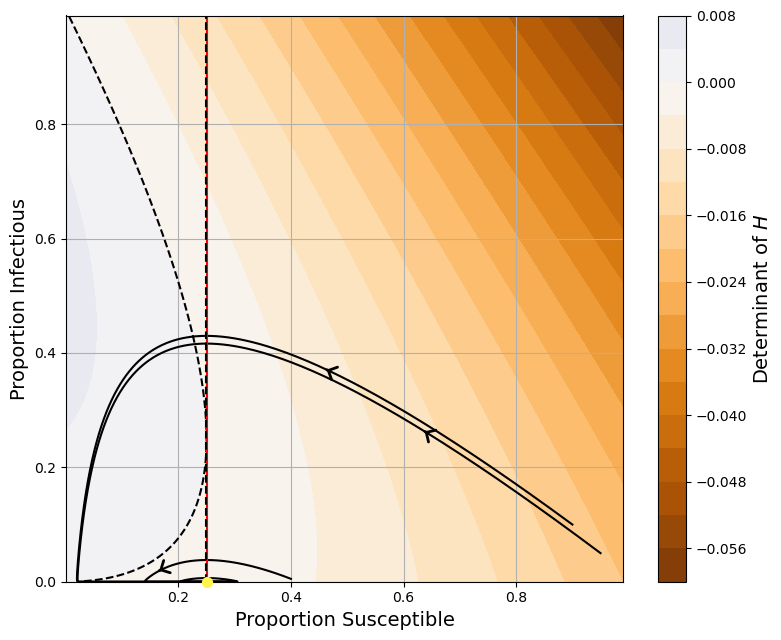}
    \caption{\label{fig:SIRnn} Phase plane of the SIR model with $R_0=4,1/\mu=80\text{ years}$ and heat map denoting the determinant of $H$. The curved dotted black line shows the $H_D=0$ manifold, the vertical dotted black line corresponds to $R_t=1$ and the vertical solid red line indicates the $S$ coordinate of the endemic equilibrium. The solid yellow dot indicates the endemic equilibrium.}
\end{figure}

Trajectories that do not cross the manifold start close to the stable EE and indicate that transient growth in fluctuations is always possible. The trajectories indicate that transient growth in perturbations is not possible just after the epidemic peak has been reached (the vertical line). More interesting behaviour can be seen when $\mu \gg 0$, as shown in Fig.~\ref{fig:SIRnnhighmu}. Now, the number of crossings is dependent on the initial state of the system, with a large initial infected population required to cross twice. The system must have a high proportion susceptible for transient growth in fluctuations to be possible.

\begin{figure}[h!]
    \centering
    \includegraphics[width=\linewidth]{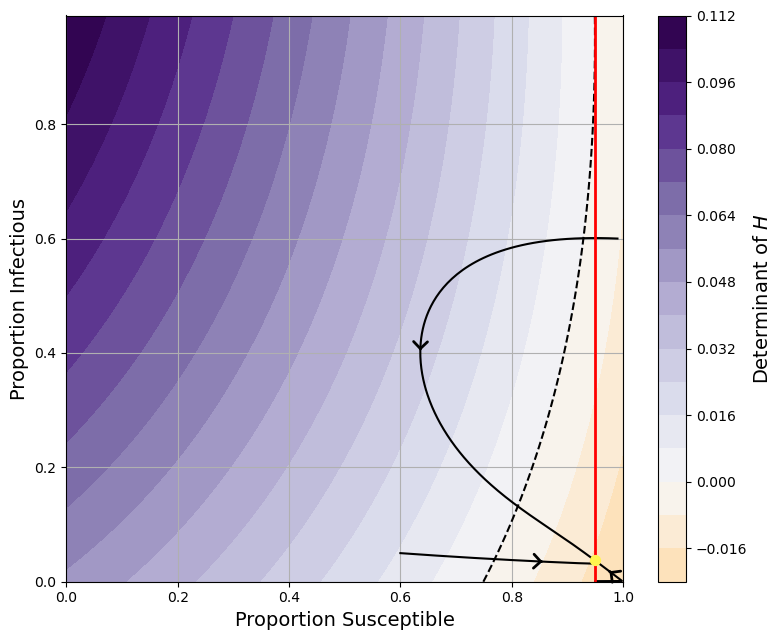}
    \caption{\label{fig:SIRnnhighmu} Phase plane of the SIR model with $R_0=4,1/\mu=1/73\text{ years}$ and heat map denoting the determinant of $H$. The curved dotted black line shows the $H_D=0$ manifold, the vertical dotted black line corresponds to $R_t=1$ and the vertical solid red line indicates the $S$ coordinate of the endemic equilibrium. The solid yellow dot indicates the endemic equilibrium.}
\end{figure}

For $R_0<1$ we observe phase planes similar to Fig.~\ref{fig:SIRnnlowR0} where a sufficiently high initial infectious population is required to have $\det H>0$. Again, this behaviour is expected from the arguments used when considering the $D=0$ manifold.

\begin{figure}[h!]
    \centering
    \includegraphics[width=\linewidth]{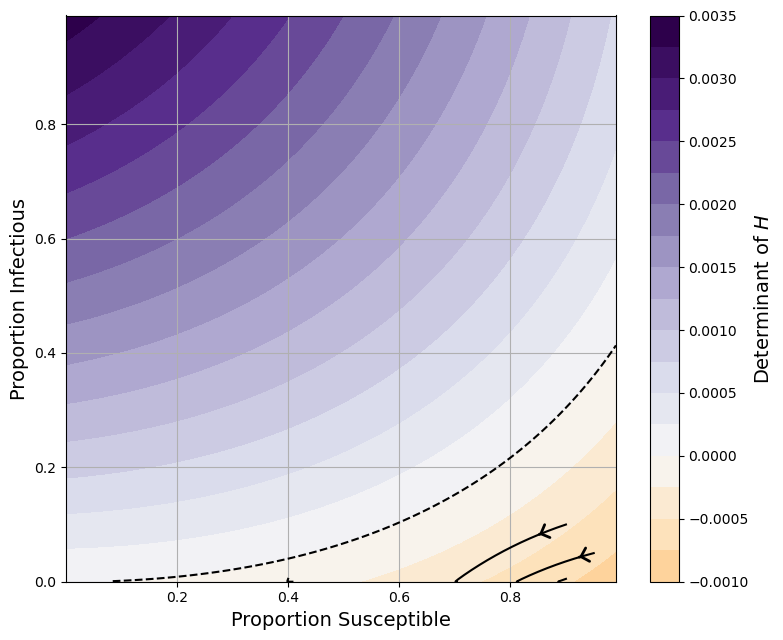}
    \caption{\label{fig:SIRnnlowR0} Phase plane of the SIR model with $R_0=0.840,1/\mu=80\text{ years}$ and heat map denoting the determinant of $H$. The curved dotted black line shows the $H_D=0$ manifold and there is no endemic equilibrium.}
\end{figure}

\newpage
\subsection{Theoretical covariance evolution}
Having analysed the geometry of the SIR model and conditions for both asymptotic and transient growth in time-dependent fluctuations from the mean-field solution, we can now investigate the behaviour of the time-dependent covariance matrix of detrended fluctuations under the linear-noise approximation. 

We simulate three outbreaks of a disease shown in Figs.~\ref{fig:SIRmeanhighR0} to~\ref{fig:SIRmeanhighS0}. These pathogen dynamics and initial conditions were chosen to best demonstrate the theoretical variance decomposition and are not supposed to represent any previous or current real-world pandemic.

\begin{figure}[h]
    \centering
    \includegraphics[width=\linewidth]{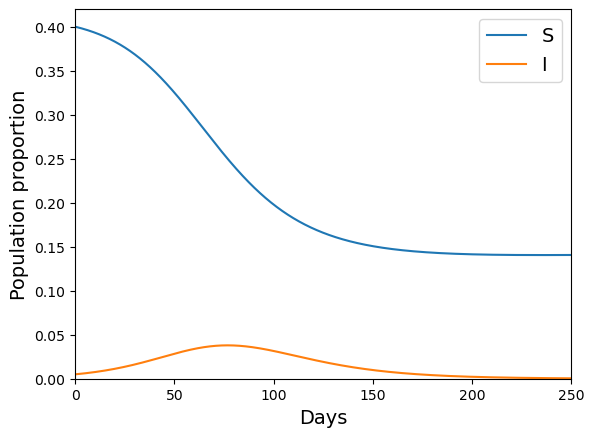}
    \caption{\label{fig:SIRmeanhighR0} Proportion susceptible and infectious for $R_0=4,1/\mu=$ 80 years and initial conditions $(S=0.4,I=0.05,R=0.595$.}
\end{figure}

\begin{figure}[h]
    \centering
    \includegraphics[width=\linewidth]{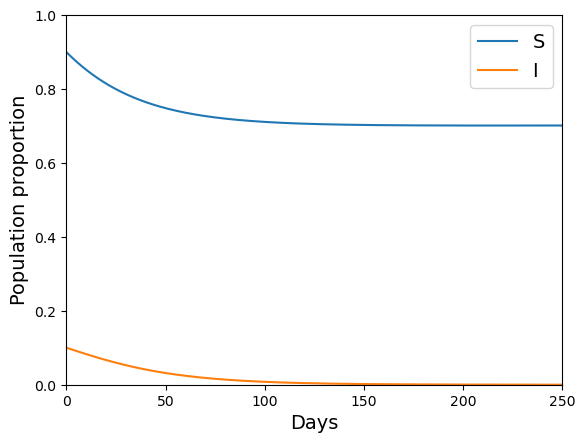}
    \caption{\label{fig:SIRmeanlowR0} Proportion susceptible and infectious for $R_0=0.840,1\mu=$ 80 years and initial conditions $(S=0.9,I=0.1,R=0$.}
\end{figure}

\begin{figure}[h]
    \centering
    \includegraphics[width=\linewidth]{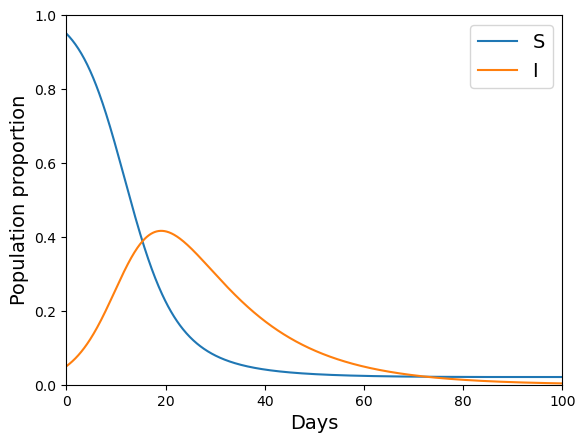}
    \caption{\label{fig:SIRmeanhighS0} Proportion susceptible and infectious for $R_0=4,1/\mu=$ 80 years and initial conditions $(S=0.95,I=0.05,R=0$.}
\end{figure}
We work with the covariance matrix in the $J$ eigenbasis (with coordinate transformation given by the matrix of eigenvectors, V) and denoting the unit vector in the $I$ direction by $e_I$, and we also define $g={V}^{*}e_I$ the projection of $I$ into the eigenbasis. Then we can consider a decomposition of the $\operatorname{Var}(I)$ dynamics as 
\begin{align}
    \operatorname{Var}(I) &= e_I^{T}V\Sigma_yV^{*}e_I=g^{*}\Sigma_yg \\
    &= \sum_{i}|g_i|^2{\Sigma_{y}}_{ii}+\sum_{i\neq j}\bar{g_i}\Sigma_{ij}g_j
\end{align}
where the second line comes from splitting the contributions into diagonal and off-diagonal elements. So we can see that the instantaneous variance in $I$ comes from the diagonal and off-diagonal terms in $\Sigma_y$ weighted by the projection of each compartment in the eigenbasis. More intuitively, we can also consider the explicit time evolution of $\operatorname{Var}(I)$. 

Returning to Eq.~\ref{eq:modal_matrix} and substituting it into the expression for $\frac{\text{d}\operatorname{Var}(I)}{\text{d}t}$ yields 
\begin{align}
    \dfrac{\text{d}\operatorname{Var}(I)}{\text{d}t} &= \dot{g}^{*}\Sigma_yg+g^{*}\dot{\Sigma}_yg+g^{*}\Sigma_y\dot{g} \\
    &=g^{*}(\Lambda\Sigma_y+\Sigma_y\Lambda^*+B_y-M\Sigma_y-\Sigma_yM^*)g \nonumber\\
    &\quad\quad\quad + g^{*}M\Sigma_yg+g^{*}\Sigma_yM^{*}g \\
    &= g^{*}(\Lambda\Sigma_y+\Sigma_y\Lambda^*+B_y)g \\ 
    &= \sum_i{2\Re(\lambda_i)|g_i|^2{\Sigma_y}_{ii}} \nonumber\\
    &\quad\quad +\sum_{i\neq j}{(\lambda_i+\bar{\lambda_j})\bar{g}_i{\Sigma_y}_{ij}g_j}+g^{*}B_yg
\end{align}
where we use the fact that $\dot{g}=\dot{V}^{*}e_I=M^{*}g$ for simplification. Now, this somewhat unwieldy equation says that the evolution of $\operatorname{Var}(I)$ has a diagonal contribution scaled by the real part of the eigenvalues and the alignment of $I$ with the eigenvectors of $J$. The off-diagonal contributions are defined similarly but with additional scaling from the alignment/rotation of the eigenvectors and then the injected noise (as expected). Thus, we expect the variance to increase as the real part of the eigenvalues becomes positive like in standard critical slowing down theory, and if the eigenvectors are rotated in the $I$ direction.

Using these equations, we can now consider the theoretical evolution of $\operatorname{Var}(I)$ for the mean-field results summarised in Figs.~\ref{fig:SIRmeanhighR0} and~\ref{fig:SIRmeanlowR0}.

\begin{figure}[h!]
    \centering
    \includegraphics[width=\linewidth]{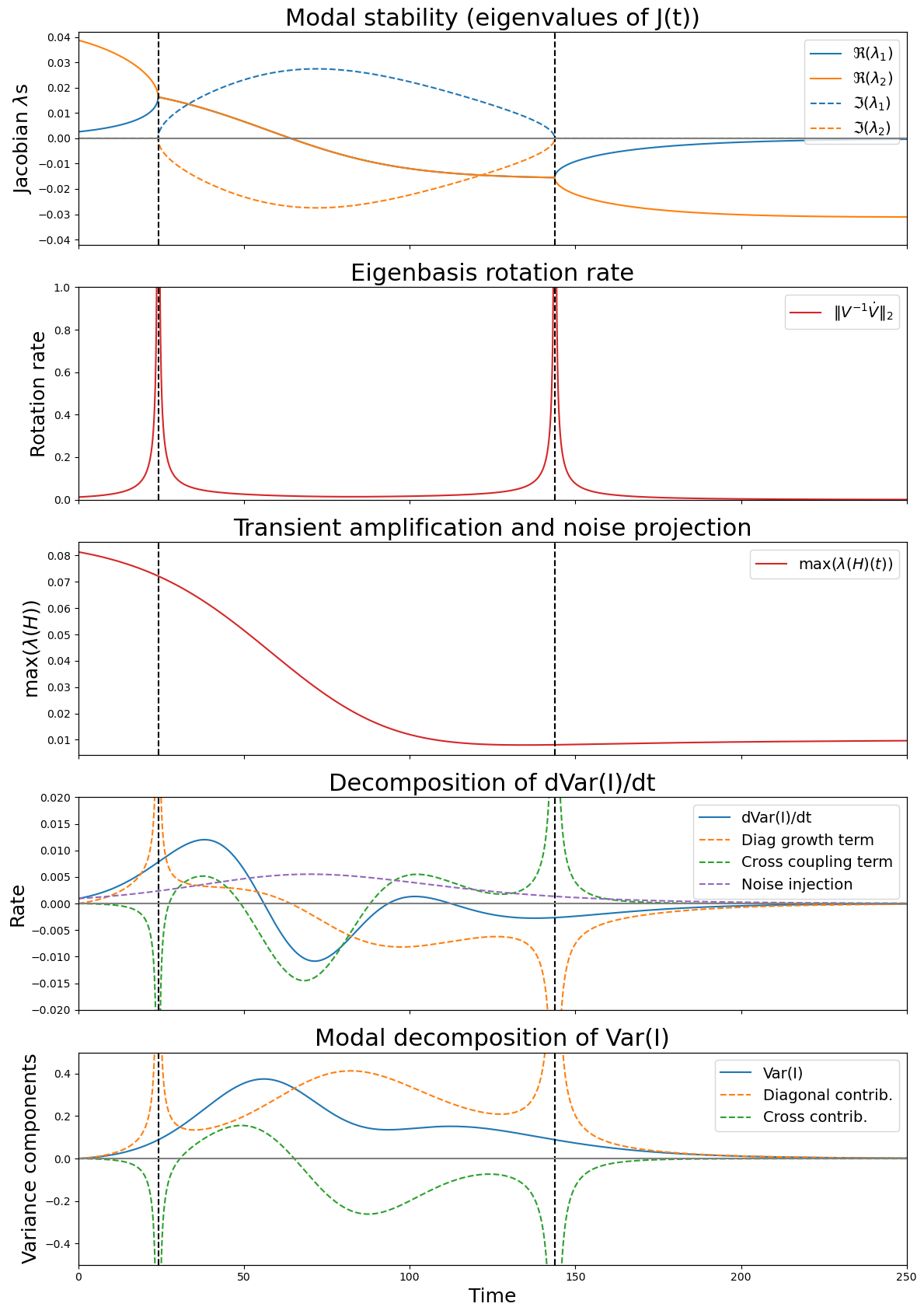}
    \caption{\label{fig:SIRfull} Theoretical eigenvalues of $J$, the rotation rate of the eigenbasis, the numerical abscissa of $H$ and the modal decomposition of $\operatorname{Var}(I)$ for the SIR model under the LNA with ($R_0=4,1/\mu=$ 80 years and initial conditions $(S=0.4,I=0.05,R=0.595)$. Vertical dashed lines show a change in eigenvalue complexity.}
\end{figure}

We observe two peaks in variance in Fig.~\ref{fig:SIRfull} (bottom plot, blue line). It is clear that the first increase is mainly caused by the initial diagonal contribution in the $J$ eigenbasis as well as a rotational contribution after the eigenvalues first become complex conjugates. This increase on approach to the change in complexity is expected from standard critical slowing down theory; we can see that the variance first peaks just before the real part of the complex conjugate Jacobian eigenvalues becomes negative. As for the second peak, this is explainable by time-dependent effects in the Jacobian eigenbasis. The fourth row in the plot shows how on approach to this second smaller peak, there is a sharp increase in the derivative due to rotational/cross-coupling terms of the eigenbasis. This indicates that, even though we would expect an asymptotic decrease in fluctuations due to the negative real part of the Jacobian eigenvalues, geometric effects can cause signals in standard early warning statistics. The second transition in eigenvalue complexity, from complex conjugate to real, does not lead to a second increase in variance as the real part of the eigenvalues is still negative and rotational effects are largely dampened by diagonal/modal effects in the eigenbasis.

We also consider the $R_0<1$ theoretical results in Fig.~\ref{fig:SIRfulllowR0}.

\begin{figure}[h!]
    \centering
    \includegraphics[width=\linewidth]{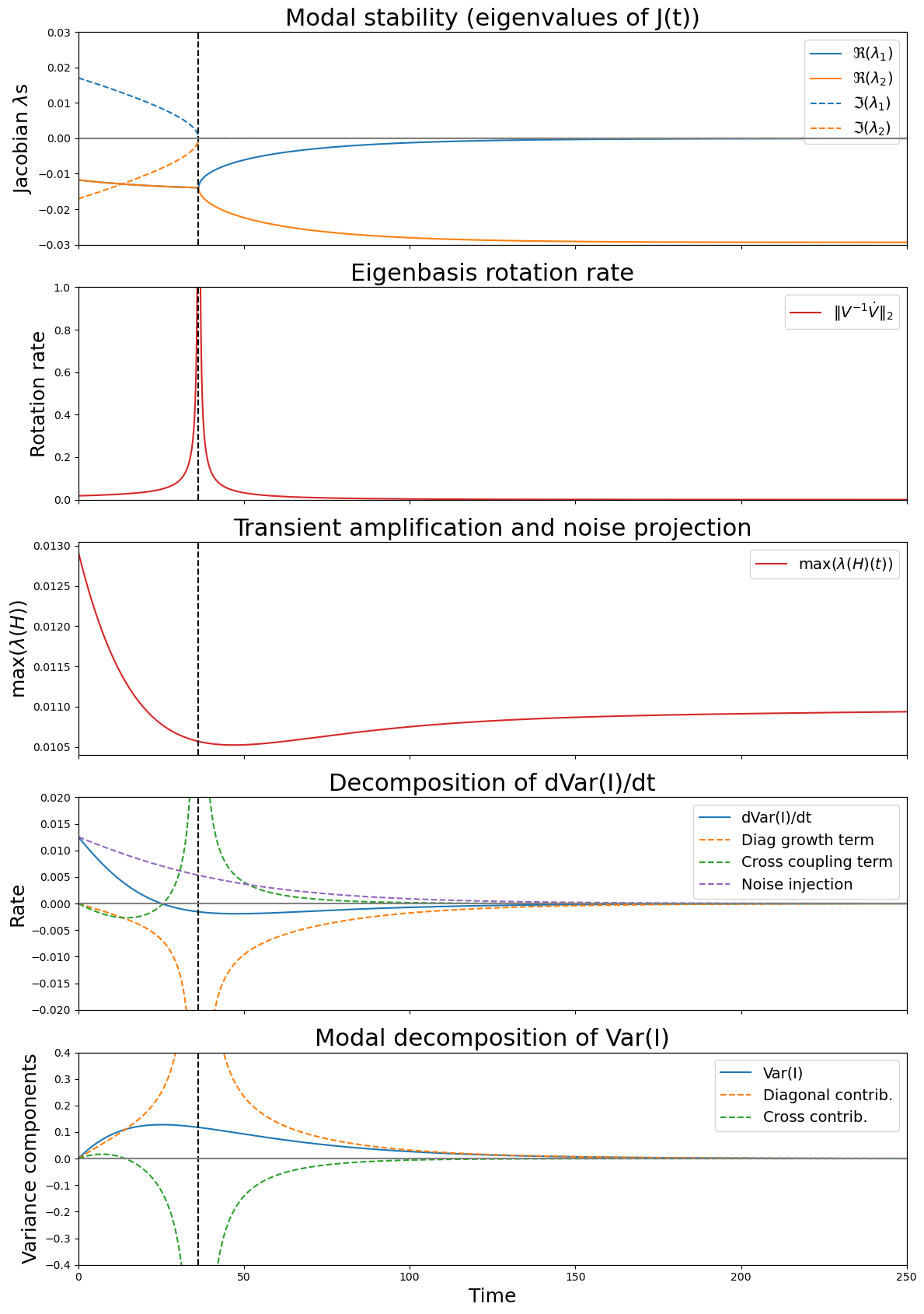}
    \caption{\label{fig:SIRfulllowR0} Theoretical eigenvalues of $J$, the rotation rate of the eigenbasis, the numerical abscissa of $H$ and the modal decomposition of $\operatorname{Var}(I)$ for the SIR model under the LNA with $R_0=0.840, 1/\mu=$ 80 years and initial conditions $(S=0.9,I=0.1,R=0)$. Vertical dashed lines show a change in eigenvalue complexity.}
\end{figure}

Unlike the $R_0=4$ case, the initial increase in variance is purely due to the noise injection (noting that the maximum eigenvalue of $H$ is positive such that transient growth in perturbations is possible).

Finally, we also consider the theoretical evolution of $\operatorname{Var}(I)$ when $R_0=4$ but with a much higher initial susceptible proportion and infectious proportion in Fig.~\ref{fig:SIRfullhighS0}, for the mean-field results shown in Fig.~\ref{fig:SIRmeanhighS0}.

\begin{figure}[h!]
    \centering
    \includegraphics[width=\linewidth]{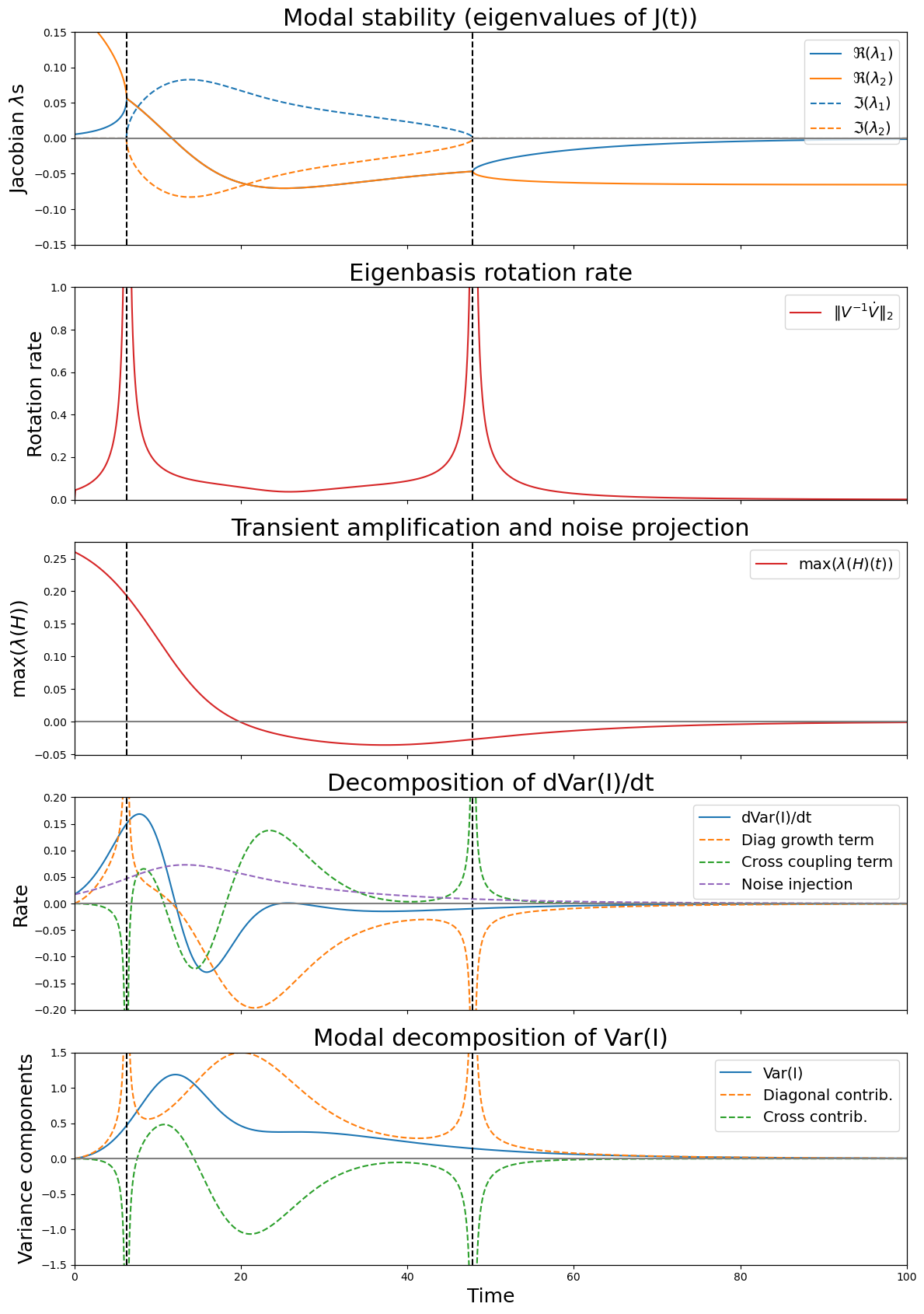}
    \caption{\label{fig:SIRfullhighS0} Theoretical eigenvalues of $J$, the rotation rate of the eigenbasis, the numerical abscissa of $H$ and the modal decomposition of $\operatorname{Var}(I)$ for the SIR model under the LNA with $R_0=4, 1/\mu=$ 80 years and initial conditions $(S=0.95,I=0.05,R=0)$. Vertical dashed lines show a change in eigenvalue complexity.}
\end{figure}

Comparing the two $R_0=4$ results with different initial conditions (Figs.~\ref{fig:SIRfull} and \ref{fig:SIRfullhighS0}), we still see qualitatively identical behaviour in the eigenvalues and variance but on a much faster time scale when we start with no recovered individuals. Further, the numerical abscissa of $H$ is negative from about $t=20$ days until the end of the simulation period. This may explain why the second peak in variance associated with the rotating Jacobian eigenbasis is slightly smaller than with a lower initial susceptible proportion. Nonetheless, the qualitative behaviour is still similar and matches what we expect from the intuition developed when analysing the $S-I$ phase plane.

\newpage
\section{Simulation results}
In order to test the limitations of the LNA when applied along a trajectory, we conducted a set of Gillespie simulations with parameter values and initial conditions as in the previous section. Each experiment used 10,000 simulations and we aggregated the jump process to daily values so that the variance could be taken between simulations at shared sampling intervals. We note that births and deaths are uncoupled in our simulations, as coupling these events adds extra correlation structure into the dynamics. The codebase can be found at \href{https://github.com/joshlooks/non-stationary-ews}{https://github.com/joshlooks/non-stationary-ews}. 

Explicitly, for $R_0=4$, the initial conditions were set as $[S_{init},I_{init},R_{init}]=[0.4,0.005,0.595]N$ and for $R_0=0.840$ as $[0.9,0.1,0]N$. The initial covariance matrix for $S,I$ fluctuations is the zero matrix as all simulations start at the same initial state.

Results for $R_0=0.840$ are presented in Fig~.\ref{fig:simslowR0}. There is clear concordance between the theoretical and empirical variance across all population sizes tested with larger population sizes leading to greater agreement. 

\begin{figure*}
    \includegraphics[width=\linewidth]{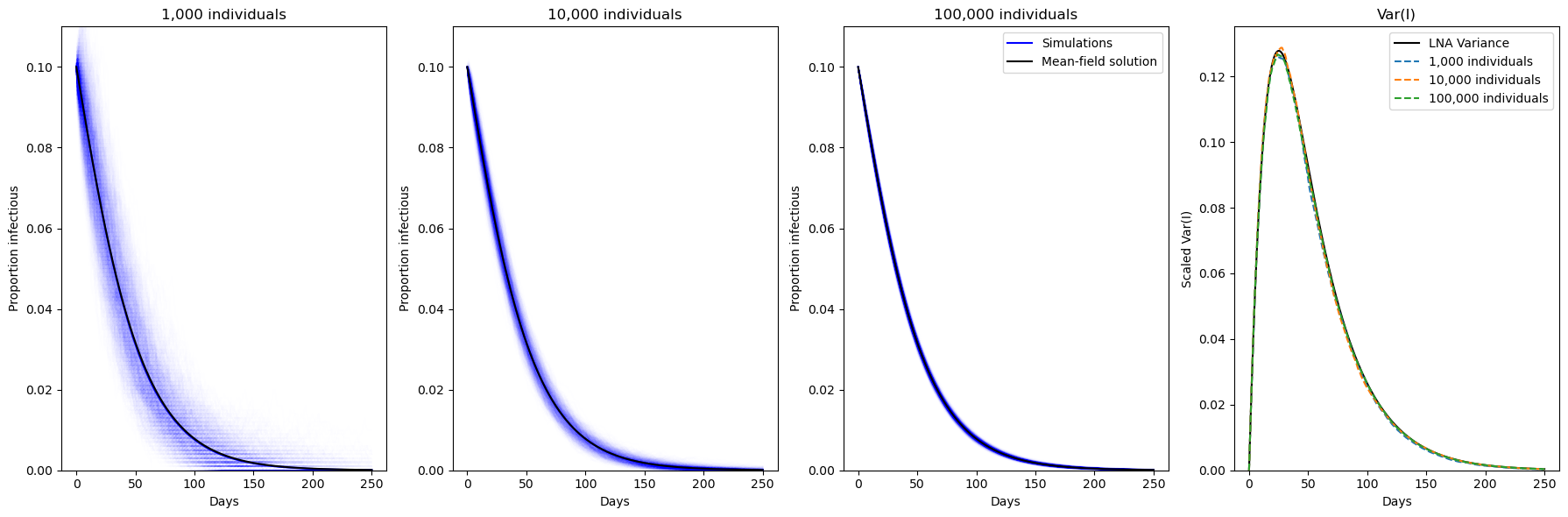}
    \caption{\label{fig:simslowR0} First three columns: proportion infectious from the results of $1,000$, $10,000$ and $100,000$ Gillespie simulations with simulation average and mean-field values overlayed for SIR model with $R_0=0.840,1/\mu=$ 80 years and initial conditions $(S=0.9,I=0.1,R=0)$. Last column: empirical and theoretical variance in the fluctuations of infectious individuals from the mean-field solution, empirical variance calculated between the simulation results at daily intervals.}
\end{figure*}

For $R_0=4$ results (Fig~.\ref{fig:simshighR0}) there is also good concordance between simulations and theory for higher population sizes. However, the $N=1,000$ simulations have a qualitatively different variance curve to what the theory predicts. The simulation average does not match the mean-field/ODE solution, likely due to the the increased impacts of stochastic effects with a lower population size. More simulations are seeing delayed take-off which is widening the variance curve and smoothing out the secondary peak expected from the theory. When we instead took the subset of 9,013 simulations that observed a `large outbreak' (which we defined as $\max{I(t)}>2I(0)$), the simulation average is closer to the mean-field solution and the first peak in the empirical variance is qualitatively similar to the theoretical prediction as seen in Fig~.\ref{fig:simshighR0quasi}. The smaller peak is due to increased correlation in the simulations, as previously analysed in \cite{Looker_identifying_2025}, but there is still a large difference in the variance curves with respect to the second theoretical peak.

\begin{figure*}
    \centering
    \includegraphics[width=\linewidth]{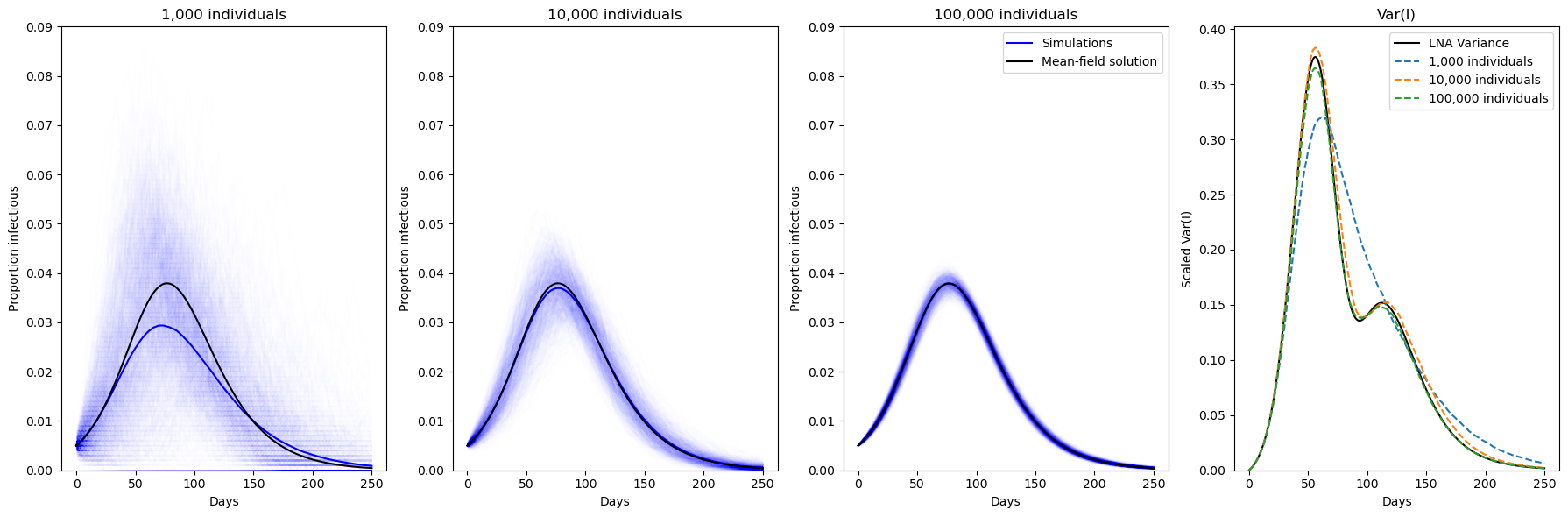}
    \caption{\label{fig:simshighR0} First three columns: proportion infectious from the results of 100,000 Gillespie simulations with population sizes of $N=1,000$, $10,000$ and $100,000$ with simulation average and mean-field values overlayed for SIR model with $R_0=4,1/\mu=$ 80 years and initial conditions $(S=0.4,I=0.05,R=0.595)$. Last column: empirical and theoretical variance in the fluctuations of infectious individuals from the mean-field solution, empirical variance calculated between the simulation results at daily intervals.}
\end{figure*}

We also tested how the size of the initial fluctuations from the mean-field solution may impact the resulting observed variance. The initial conditions for each compartment were sampled from independent truncated discretised Normal distributions, with expected total population size of $N=10,000$ and $\sigma^2\in[1000\times0.1^2,1000\times0.2^2,1000\times0.3^2]$ and the mean of each compartment equal to the mean-field initial conditions. The minimum number of initial susceptibles, infecteds and recovereds were set as $[1,1,0]$ for the $R_0=0.84$ simulations and $[1,10,1]$ for the $R_0=4$ simulations. We note that we would only expect the lower truncation to be reached in the $R_0=4$ simulations and only in the infectious compartment due to the small mean number of initial infectious. As the truncation and discretisation changed the initial variance and mean (of infectious individuals), when solving the ODE for the theoretical results we took the initial condition to be the empirical mean and variance at day 0. These results are summarised in Fig.~\ref{fig:simslownormR0} and Fig.~\ref{fig:simshighnormR0}, respectively.

\begin{figure*}
    \centering
    \includegraphics[width=\linewidth]{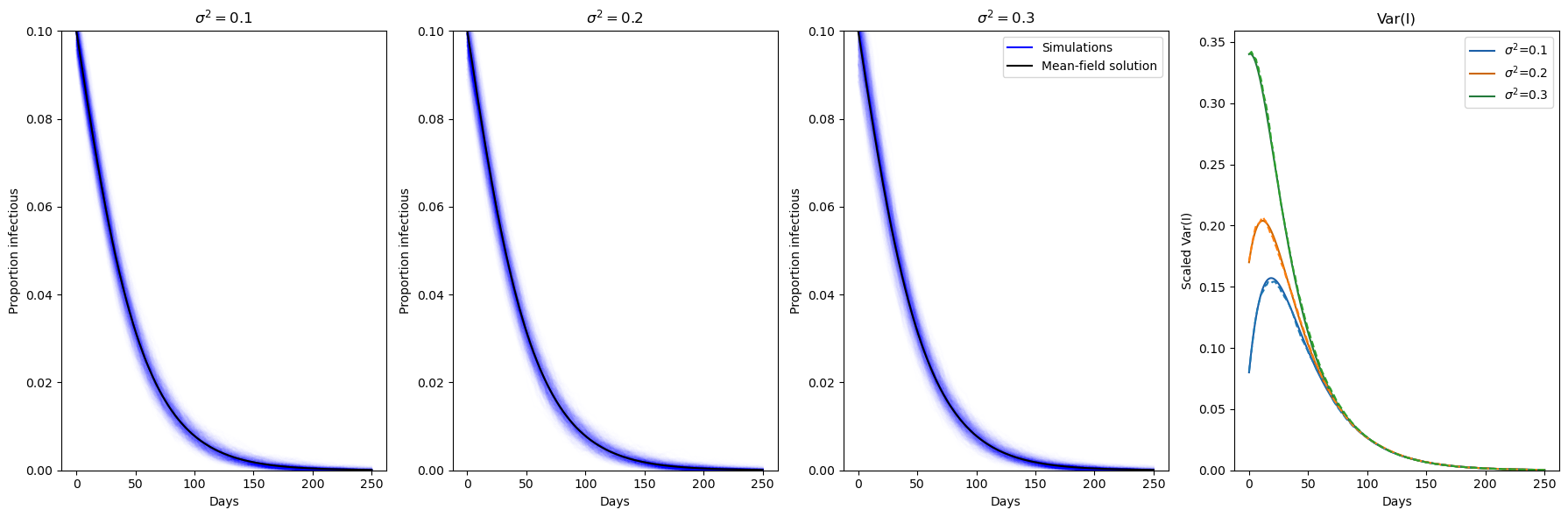}
    \caption{\label{fig:simslownormR0} First three columns: proportion infectious from the results of $100,000$ Gillespie simulations with simulation average and mean-field values overlayed for SIR model with $R_0=0.840,1/\mu=$ 80 years and initial conditions given by a truncated normal distribution with mean $(S=0.9,I=0.1,R=0)$ and variances given by $\sigma^2/N=0.1^2,0.2^2$ and $0.3^2$. Last column: empirical and theoretical variance in the fluctuations of infectious individuals from the mean-field solution, empirical variance calculated between the simulation results at daily intervals. Empirical variance is shown in dashed lines and theoretical in solid lines.}
\end{figure*}

\begin{figure*}
    \centering
    \includegraphics[width=\linewidth]{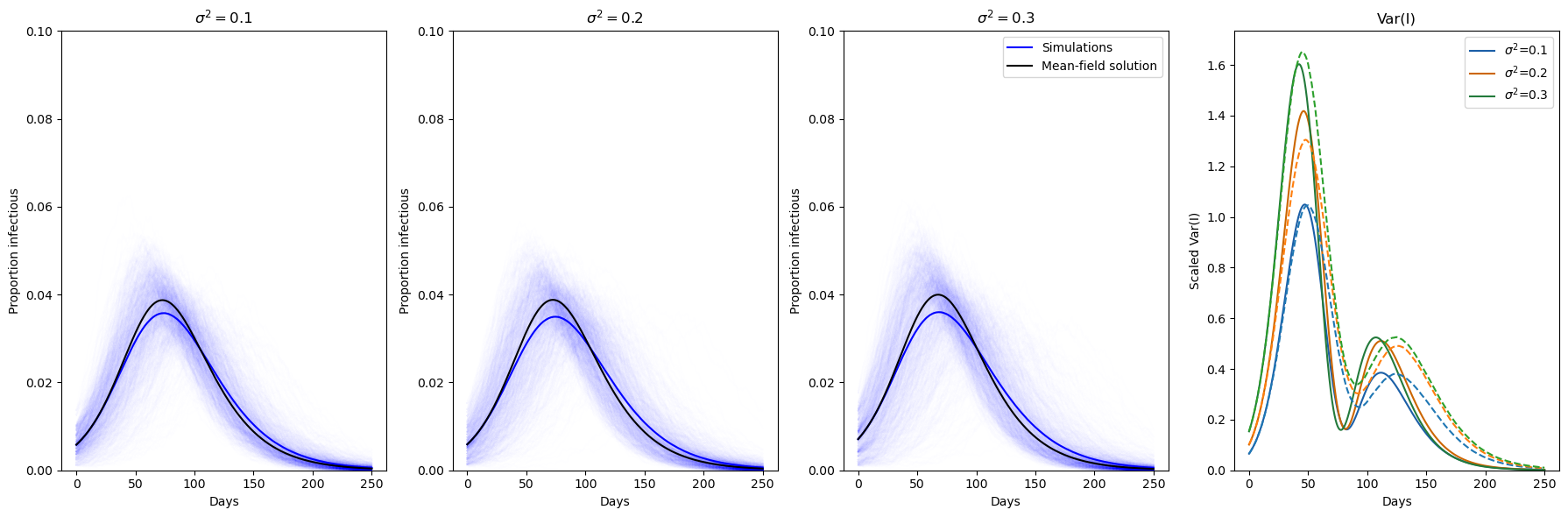}
    \caption{\label{fig:simshighnormR0} First three columns: proportion infectious from the results of $100,000$ Gillespie simulations with simulation average and mean-field values overlayed for SIR model with $R_0=4,1/\mu=$ 80 years and initial conditions given by a truncated normal distribution with mean $(S=0.4,I=0.05,R=0.595)$ and variances given by $\sigma^2/N=0.1^2,0.2^2$ and $0.3^2$. Last column: empirical and theoretical variance in the fluctuations of infectious individuals from the mean-field solution, empirical variance calculated between the simulation results at daily intervals. Empirical variance is shown in dashed lines and theoretical in solid lines.}
\end{figure*}

The LNA clearly matches the empirical variance curve for the simulations with $R_0=0.84$ for all initial variances tested (Fig.~\ref{fig:simshighnormR0}). This may also be due to how close the simulations are to the mean-field solution regardless of the starting conditions. However, the simulations with variance in initial conditions (empirical initial variance) and $R_0=4$ show the same qualitative behaviour as predicted by the LNA but with later secondary peaks. This is likely because the LNA does not capture how skewed the truncated normal distribution is for the initial infected. The mean-field solution has a higher and slightly later peak than the average of the simulations, skewing the residuals distribution and thus the empirical variance. Nonetheless, the qualitative similarities between the predicted and empirical variance curves indicate that the LNA and the time-dependent variance can be used in an EWS context.

We also changed the initial sampling distribution to a Multinomial distribution with means equal to the mean-field initial conditions. Explicitly $[S_{init},I_{init},R_{init}]\sim\text{Multinomial}(N,p)$ where $p=(p_S,p_I,p_R)$ are the expected initial proportions in each compartment (the mean-field initial conditions). The initial covariance matrix for the fluctuations was then set using a Normal distribution to the Multinomial distribution such that $\Sigma_0=(\operatorname{diag}(p)-pp^T)/N$. The results are presented in Appendix~\ref{app:further_sims} in Fig.~\ref{fig:simslowrandR0} and Fig.~\ref{fig:simshighrandR0} for the $R_0=0.840$ and $R_0=4$ results respectively. We again see that the empirical variance matches the LNA predictions.

We note that in both sets of results with Multinomial initial conditions, and especially so for the $R_0=4$ simulations, it was required to use the full system representation with all three compartments when calculating the covariance matrix evolution. This is because starting with $R_{init} \neq 0$ leads to a possibly high correlation between the $I$ and $R$ compartments which then follows through to the evolution of $\operatorname{Var}(I)$.

\newpage
\section{Discussion}
We have investigated the evolution of the time-dependent covariance matrix of fluctuations in linearised dynamical systems. Our work suggests that the geometry of such systems, in terms of the spectra and pseudo-spectra, provides a mechanistic explanation for transient critical-like behaviour observed away from dynamic equilibria.

We have shown that the covariance evolution of the fluctuations can be understood from the spectra and pseudo-spectra of the non-autonomous Jacobian. The intuitive relationship between the criticality of the spectra and `time-dependent' critical slowing down was derived, and further effects from the evolution of the eigenbasis itself were discussed. However, there are possible limitations when applying this theory to data. When estimating the distribution of residuals around a fixed point of a dynamical system, as in standard EWS analysis to anticipate a bifurcation, ergodic theory suggests that we can use the time-average of the data. This is no longer true in the non-autonomous/time-dependent case. Indeed ergodic theory has also failed with simulation- or data-driven EWS analysis around a steady-state \cite{lenton_early_2012,ODea_disentangling_2019,gama_dessavre_problem_2019}. Nonetheless, it may be possible to robustly estimate the fluctuation statistics if the mean-field solution is sufficiently well-behaved or if time-series from similar processes exist \cite{ODea_disentangling_2019,gama_dessavre_problem_2019,Looker_identifying_2025}.

We note here that our framework does not propose a novel EWS metric, but adds to the literature on the underlying geometric mechanisms that generate CSD-like signatures away from bifurcations. The non-normality itself may not change on approach to a tipping point, but a system trajectory may evolve in areas of the phase-plane where the local Jacobian exhibits varying degrees of non-normality. To leverage this geometric understanding in practice, the non-normality must be measurable from observational data with recent advances in the literature indicating that this is possible \cite{Troude_pseudo_2025,Saiprasad_non_normal_inference_2026}. Considering these indicators alongside traditional EWS could help modellers identify whether an observed statistical signal is caused by true asymptotic critical slowing down, or transient amplification driven by time-varying non-normal geometry.

Another clear limitation of the linearised theory is that fluctuations will only satisfy the linearised dynamics if they remain sufficiently close to the mean-field solution. This becomes a failure point near a tipping or bifurcation threshold (such as $R_0 \to 1$). As the system approaches criticality, the leading eigenvalue of the Jacobian approaches zero and the linear restoring force vanishes. Because the continuous-time Lyapunov equation relies solely on this linearised drift, it predicts that the covariance will diverge to infinity. In the true stochastic system, higher-order nonlinear drift terms (which the linearised approximation truncates) become the dominant restoring forces and bound the actual variance \cite{oregan_theory_2013,Papapicco_likelihood_2026,kuehn_mathematical_2011,Donovan_moment_2024}. Therefore, while our linear evolution of covariance perfectly captures transient, geometry-driven EWS in stable regimes, it is not quantitatively faithful to the actual fluctuation dynamics in the immediate asymptotic vicinity of a true bifurcation.

Further, in the $R_0>1$ parameter range, where analysis of the variance in the infectious population is arguably most useful, this was only true in the large population limit and was also dependent on the initial fluctuation distribution. We also note that population sizes above $N=100,000$ were not considered as fluctuations would track the mean-field solution closely (as seen in the $N=100,000$ results) and deterministic compartment modelling is likely sufficient to anticipate possible bifurcations or transient behaviour.

Apart from the $N=1,000$ with $R_0=4$ simulations, the initial prevalence was high enough that stochastic elimination did not skew the fluctuations distribution. There are papers in the literature that have analysed how lower prevalence and stochastic die-out lead to less normally distributed residuals and correlated results. Future work could consider the effects of using the quasi-stationary distribution of the continuous fluctuation process. We note that it is common in epidemic modelling to tie birth and death processes together in stochastic simulations (i.e. in the Gillespie algorithm, a death is always replaced by a new susceptible birth at the same time). This introduces correlation structures that might not be present in real-world systems; however, the effect is proportional to $\mu$ and so has minimal impact on models of most infection systems. Future work could also investigate how different modelling choices of transition functions and the noise-covariance matrix impact observed covariance trends.

This study can also help explain recent results in the epidemiological EWS literature where bifurcations in COVID-19 time series were anticipated using EWS analysis \cite{Looker_identifying_2025,dablander_overlapping_2022,proverbio_performance_2022}. The efficacy of EWS signals was found to vary with COVID-19 waves and whether authors were anticipating extinction or a re-infection wave. Our analysis indicates that this is likely due to the geometry of the epidemic system around $R_t=1$ obscuring possible $R_0=1$ bifurcations from intervention measures and new variants. This also explains why the overlapping time scales of $R_t<1$ as the epidemic goes through the peak can reduce the early warning of increasing infection from the introduction of new variants or relaxation of intervention measures. Future work should consider how transient changes in Jacobian spectral criticality may influence signal trends when trying to anticipate bifurcations in general dynamical systems. We expect the assumption of a steady-state mean-field would vary in applicability depending on the modelled infection due to differences in time-scales.

In this paper, we have investigated the linearised dynamics of fluctuations away from the mean-field solutions of non-normal dynamical systems. We derived relationships showing how the geometric properties of the system govern the evolution of the covariance matrix under the linear noise approximation. These relationships were then applied to the SIR infectious disease model to show that related early warning signals can anticipate epidemic waves and extinction even without the presence of a bifurcation, with these results extendable to all non-normal dynamical systems of interest.

\section*{Data Availability}
The code repository for the study is available at: \href{https://github.com/joshlooks/non-stationary-ews}{https://github.com/joshlooks/non-stationary-ews}. Archived code available at: \href{https://doi.org/10.5281/zenodo.18336282}{https://doi.org/10.5281/zenodo.18336282}.

\section*{Acknowledgements}
JL is supported by the Engineering and Physical Sciences Research Council through the Mathematics of Systems II Centre for Doctoral Training at the University of Warwick (reference EP/S022244/1). The funders played no role in the study design, data collection and analysis, decision to publish, or preparation of the manuscript.

\section*{Author contributions}
\textbf{Joshua Looker:} Conception, Data curation, Formal analysis, Investigation, Methodology, Software, Validation, Visualisation, Writing - Original Draft, Writing - Review \& Editing.
\noindent \textbf{Kat Rock:} Supervision, Visualisation, Writing - Review \& Editing.
\noindent \textbf{Louise Dyson:} Supervision, Visualisation, Writing - Review \& Editing.

\clearpage
\newpage
\onecolumngrid

\appendix
\section{Deriving the fluctuation bounds}\label{app:gini}
In this appendix, we explicitly derive the bounds on the fluctuations given in the main text. 
Beginning with the lower limit and defining $g(t)=\Phi(t,t_0)x$ for unit vector $x$ then

\begin{align}
    g(t+h) =\Phi(t+h,t)g(t) &=(I+hJ(t)+o(h))g(t) \\
    \dfrac{||g(t+h)||-||g(t)||}{h} & \ge \inf_{||v||=1}{\dfrac{||I+hJv||-1}{h}}||g(t)||\\
    &\quad\quad\quad+o(1)||g(t)|| \\
    \implies \dfrac{\mathrm{d}}{\mathrm{d}t^-}\ln||g(t)|| &\ge \mu^*(J)
\end{align}

where the last line comes from taking $\liminf_{h\rightarrow 0^+}$. Similarly, we can derive that 

\begin{equation}
    \dfrac{\mathrm{d}}{\mathrm{d}t^+}\ln{||g(t)||} \le \mu(J)
\end{equation}

thus 

\begin{equation}
    \exp{\left(\int_{t_0}^{t}{\mu^*(J(s))ds} \right)} \le ||\Phi(t,t_0)|| \le \exp{\left(\int_{t_0}^{t}{\mu(J(s))ds} \right)}
\end{equation}

and for $||\cdot||=||\cdot||_2$ we have

\begin{equation}
    \exp{\left(\int_{t_0}^{t}{\lambda_{\min}\left(\dfrac{J+J^T}{2}\right)ds} \right)} \le ||\Phi(t,t_0)||_2
\end{equation}

and

\begin{equation}
    ||\Phi(t,t_0)||_2 \le \exp{\left(\int_{t_0}^{t}{\lambda_{\max}\left(\dfrac{J+J^T}{2}\right)ds} \right)}
\end{equation}
\newpage
\section{Further simulation results}\label{app:further_sims}
\begin{figure*}[h]
    \centering
    \includegraphics[width=\linewidth]{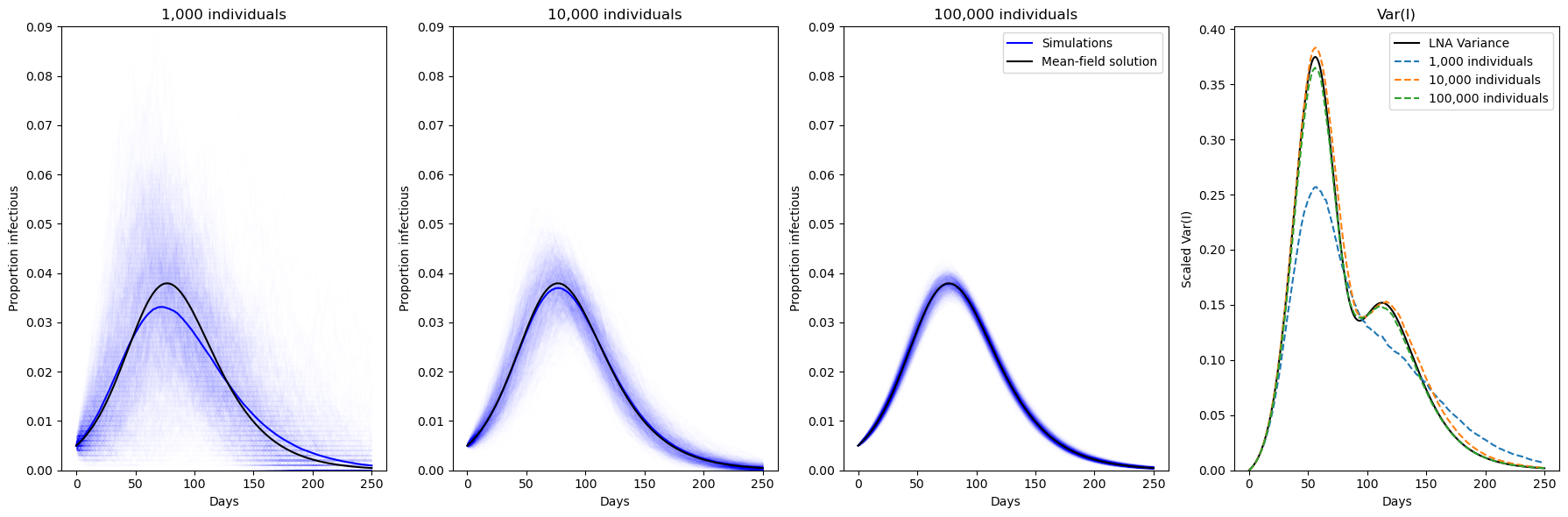}
    \caption{\label{fig:simshighR0quasi} First three columns: proportion infectious from the results of 100,000 Gillespie simulations with population sizes of $N=1,000$, $10,000$ and $100,000$ with simulation average and mean-field values overlayed for SIR model with pathogen dynamics in row 1 of Table~\ref{tab:params} and initial conditions $(S=0.4,I=0.05,R=0.595)$. Last column: empirical and theoretical variance in the fluctuations of infectious individuals from the mean-field solution, empirical variance calculated between the simulation results at daily intervals. The quasi-stationary distribution was used for the $N=1,000$ results corresponding to the subset of simulations that did not die out before the end of the simulation period.}
\end{figure*}

\begin{figure*}[h]
    \centering
    \includegraphics[width=\linewidth]{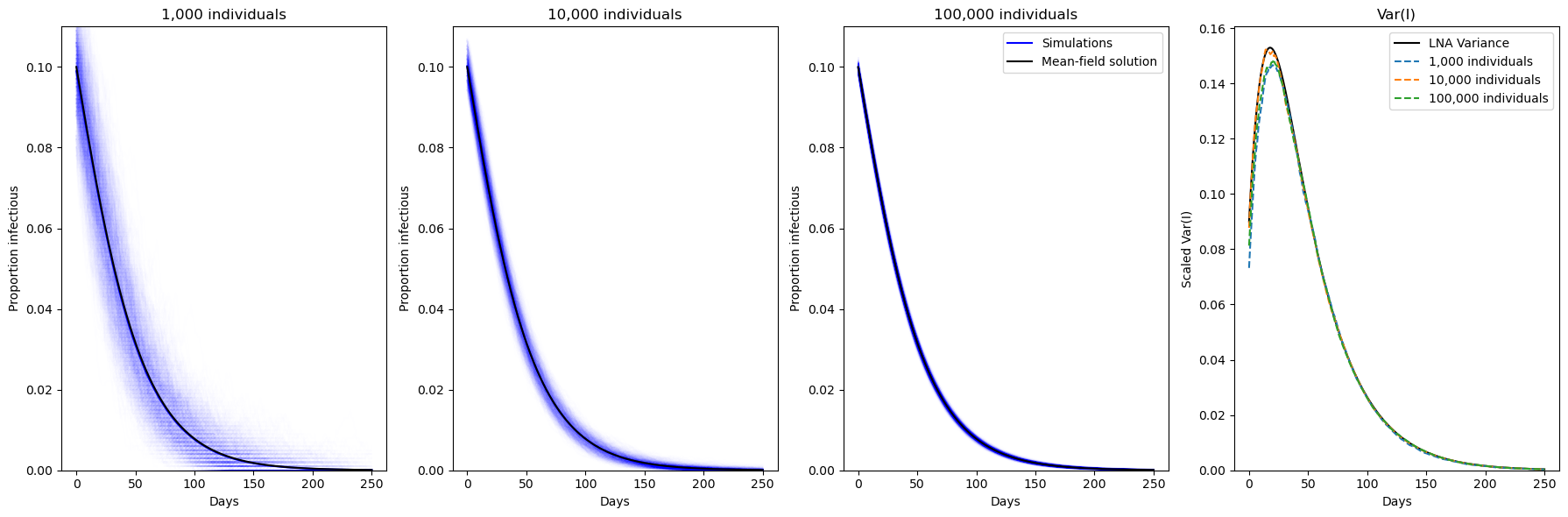}
    \caption{\label{fig:simslowrandR0} First three columns: proportion infectious from the results of $100,000$ Gillespie simulations with population sizes of $N=1,000$, $10,000$ and $100,000$ with simulation average and mean-field values overlayed for SIR model with pathogen dynamics in row 3 of Table~\ref{tab:params} and initial conditions given by a multinomial distribution with means $(S=0.9,I=0.1,R=0)$. Last column: empirical and theoretical variance in the fluctuations of infectious individuals from the mean-field solution, empirical variance calculated between the simulation results at daily intervals.}
\end{figure*}

\begin{figure*}[h]
    \centering
    \includegraphics[width=\linewidth]{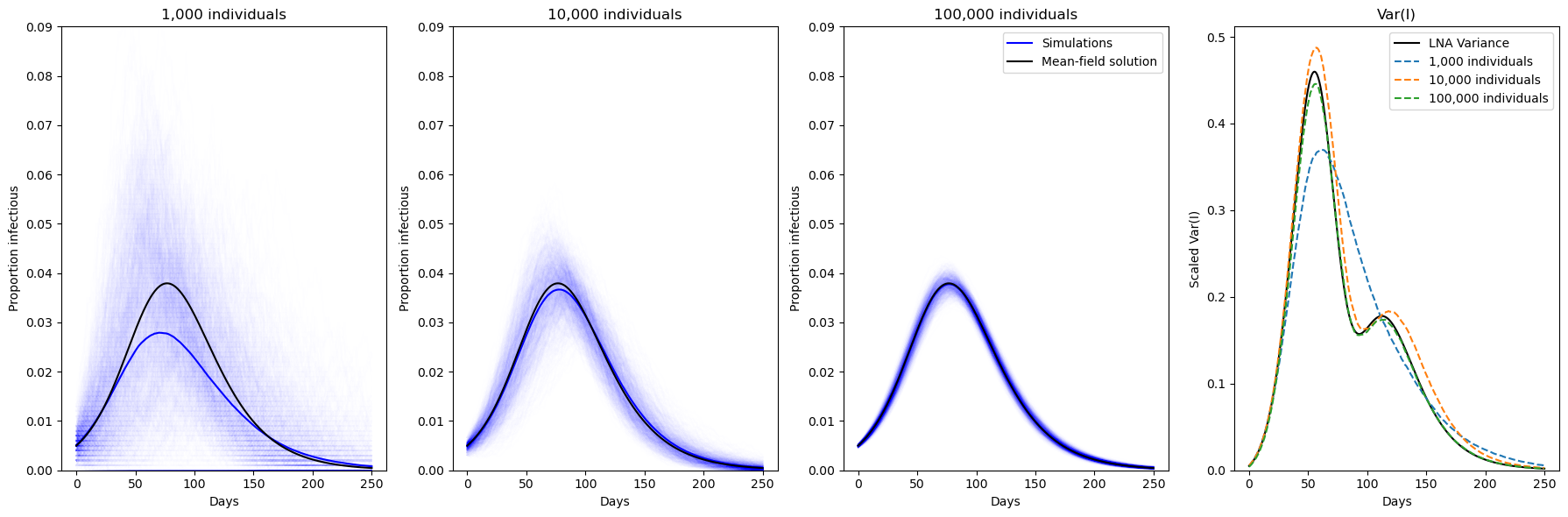}
    \caption{\label{fig:simshighrandR0} First three columns: proportion infectious from the results of $100,000$ Gillespie simulations with population sizes of $N=1,000$, $10,000$ and $100,000$ with simulation average and mean-field values overlayed for SIR model with pathogen dynamics in row 1 of Table~\ref{tab:params} and initial conditions given by a multinomial distribution with means $(S=0.4,I=0.05,R=0.595)$. Last column: empirical and theoretical variance in the fluctuations of infectious individuals from the mean-field solution, empirical variance calculated between the simulation results at daily intervals.}
\end{figure*}
\end{document}